\newif\ifprocs
\procstrue
\procsfalse 

\newif\ifarxiv
\arxivtrue

\newif\ifcomments
\commentstrue
\commentsfalse 

\ifprocs 
\documentclass[sigconf,screen]{acmart}
\else
\documentclass[11pt,USenglish]{article}
\usepackage[margin=1.0in]{geometry}
\usepackage[T1]{fontenc}
\usepackage[utf8]{inputenc}
\usepackage[english]{babel}
\usepackage{authblk}
\usepackage{chngcntr}
\counterwithin{equation}{section}
\usepackage{amssymb}
\fi

\usepackage[normalem]{ulem}
\usepackage{enumerate}
\usepackage{mwe}
\usepackage{caption}
\usepackage{xcolor}
\usepackage{algorithm}
\usepackage[noend]{algpseudocode}

\usepackage{amsmath}
\usepackage{graphicx}
\ifarxiv\else
\graphicspath{{Figures/}}
\fi
\usepackage[export]{adjustbox}
\usepackage{thmtools}
\usepackage{amsthm}
\usepackage{thm-restate}
\usepackage{xspace}

\ifprocs\else
\usepackage[colorlinks,linkcolor=blue,citecolor=blue,filecolor=blue,urlcolor=blue]{hyperref}
\fi

\newcommand{\defproblem}[2]{
\bigskip
\begin{center}
\noindent\fbox{
	\begin{minipage}{.96\linewidth}
	\textbf{#1}
	
	\smallskip
		{#2}
	\end{minipage}
}
\end{center}
\medskip
}

\newtheorem{theorem}{Theorem}[section]
\newtheorem{lemma}[theorem]{Lemma}

\newtheorem{definition}[theorem]{Definition}

\newtheorem{openq}[theorem]{Open Question}
\newtheorem{observation}[theorem]{Observation}
\newtheorem{fact}[theorem]{Fact}
\newtheorem{corollary}[theorem]{Corollary}

\theoremstyle{plain}
\newtheorem{claim}[theorem]{Claim}

\ifprocs
\else

\makeatletter
\newtheorem*{rep@theorem}{\rep@title}
\newcommand{\newreptheorem}[2]{%
\newenvironment{rep#1}[1]{%
 \def\rep@title{#2 \ref{##1}}%
 \begin{rep@theorem}}%
 {\end{rep@theorem}}}
\makeatother
\newreptheorem{theorem}{Theorem}

\makeatletter
\newtheorem*{rep@corollary}{\rep@title}
\newcommand{\newrepcorollary}[2]{%
\newenvironment{rep#1}[1]{%
 \def\rep@title{#2 \ref{##1}}%
 \begin{rep@corollary}}%
 {\end{rep@corollary}}}
\makeatother
\newreptheorem{corollary}{Corollary}

\newreptheorem{lemma}{Lemma}

\fi

\def\compactify{\itemsep=0pt \topsep=0pt \partopsep=0pt \parsep=0pt}

\newcommand{\colnote}[3]{\textcolor{#1}{$\ll$\textsf{#2}$\gg$}}

\ifcomments
\newcommand{\rnote}[1]{\colnote{red}{#1--Robi}{RK}}
\newcommand{\anote}[1]{\colnote{olive}{#1--Amir}{AA}}
\newcommand{\onote}[1]{\colnote{blue}{#1--Ohad}{OT}}
\else
\newcommand{\rnote}[1]{}
\newcommand{\anote}[1]{}
\newcommand{\onote}[1]{}
\fi

\newcommand{\ProblemName}[1]{\textsf{#1}}
\newcommand{\MF}{\ProblemName{Max-Flow}\xspace}
\newcommand{\MC}{\ProblemName{Min-Cut}\xspace}
\newcommand{\GMC}{\ProblemName{Global-Min-Cut}\xspace}
\newcommand{\EC}{\ProblemName{Edge-Connectivity}\xspace}

\newcommand{\LMC}{\ProblemName{Latest Min-Cut}\xspace}
\newcommand{\LMwC}{\ProblemName{Latest Min$_{2w}$-Cut}\xspace}
\newcommand{\expanders}{\textsc{Expanders-Guided Querying}\xspace}
\newcommand{\pC}{\textsc{Isolating-Cuts}\xspace}

\newcommand{\APMF}{\ProblemName{All-Pairs Max-Flow}\xspace}
\newcommand{\APEC}{\ProblemName{All-Pairs Edge Connectivity}\xspace}
\newcommand{\APSP}{\ProblemName{All-Pairs Shortest-Path}\xspace}
\newcommand{\SPath}{\ProblemName{Shortest-Path}\xspace}

\newcommand{\lucky}{isolated\xspace}
\newcommand{\GHT}{Gomory-Hu tree\xspace}

\newcommand{\GHEPT}{\ProblemName{GH-Equivalent Partition Tree}\xspace}

\DeclareMathOperator{\cdeg}{cdeg} 
\DeclareMathOperator{\vol}{vol} 

\let\poly\relax
\DeclareMathOperator{\poly}{poly}

\newcommand\eps{\varepsilon}
\renewcommand\epsilon{\varepsilon}
\newcommand\tO{\ensuremath{\tilde O}}

\newcommand{\Raecke}{R\"{a}cke\xspace} 

\newcommand{\T}{\mathcal{T}}
\newcommand{\TG}{\mathcal{T}^*} 

\providecommand{\set}[1]{{\{#1\}}}
\providecommand{\card}[1]{\lvert#1\rvert}

\begin{document}

\ifprocs
\title{Subcubic Algorithms for Gomory-Hu Tree in Unweighted Graphs}
\titlenote{Full version available at \href{https://arxiv.org/abs/2012.10281}{arXiv:2012.10281}.}
\else
\title{Subcubic Algorithms for Gomory-Hu Tree in Unweighted Graphs}
\fi

\ifprocs
\author{Amir Abboud}
\authornote{Work partially done at the IBM Almaden Research Center, US.}
\affiliation{%
  \institution{Weizmann Institute of Science, Israel}
  \city{}
  \country{}
}
\email{amir.abboud@weizmann.ac.il}

\author{Robert Krauthgamer}
\authornote{Work partially supported by ONR Award N00014-18-1-2364,
  the Israel Science Foundation grant \#1086/18,
  and a Minerva Foundation grant.
}
\affiliation{%
  \institution{Weizmann Institute of Science, Israel}
  \city{}
  \country{}
}
\email{robert.krauthgamer@weizmann.ac.il}

\author{Ohad Trabelsi}
\affiliation{%
  \institution{Weizmann Institute of Science, Israel}
  \city{}
  \country{}
}
\email{ohad.trabelsi@weizmann.ac.il}

\else

\author[1]{Amir Abboud%
  \thanks{Work partially done at the IBM Almaden Research Center.
    Email: \texttt{amir.abboud@weizmann.ac.il}
  }}
\author[1]{Robert Krauthgamer\thanks{Work partially supported by ONR Award N00014-18-1-2364, the Israel Science Foundation grant \#1086/18, and a Minerva Foundation grant.
    Email: \texttt{robert.krauthgamer@weizmann.ac.il}
  }}
\author[1]{Ohad Trabelsi%
  \thanks{Email: \texttt{ohad.trabelsi@weizmann.ac.il}
  }}
\affil[1]{Weizmann Institute of Science}

\fi

\ifprocs
\setcopyright{acmlicensed}
\acmPrice{15.00}
\acmDOI{10.1145/3406325.3451073}
\acmYear{2021}
\copyrightyear{2021}
\acmSubmissionID{stoc21main-p262-p}
\acmISBN{978-1-4503-8053-9/21/06}
\acmConference[STOC '21]{Proceedings of the 53rd Annual ACM SIGACT Symposium on Theory of Computing}{June 21--25, 2021}{Virtual, Italy}
\acmBooktitle{Proceedings of the 53rd Annual ACM SIGACT Symposium on Theory of Computing (STOC '21), June 21--25, 2021, Virtual, Italy}
\fi 

\ifprocs\else
\maketitle
\thispagestyle{empty}
\setcounter{page}{0}
\fi

\begin{abstract}
Every undirected graph $G$ has a (weighted) \emph{cut-equivalent tree} $T$, commonly named after Gomory and Hu who discovered it in 1961.
Both $T$ and $G$ have the same node set, and for every node pair $s,t$,
the minimum $(s,t)$-cut in $T$ is also an exact minimum $(s,t)$-cut in $G$.

We give the first subcubic-time algorithm that constructs such a tree
for a simple graph $G$ (unweighted with no parallel edges). 
Its time complexity is $\tilde{O}(n^{2.5})$, for $n=|V(G)|$;
previously, only $\tilde{O}(n^3)$ was known, 
except for restricted cases like sparse graphs.
Consequently, we obtain the first algorithm
for \APMF in simple graphs that breaks the cubic-time barrier.

Gomory and Hu compute this tree using $n-1$ queries to (single-pair) \MF; the new algorithm can be viewed as a fine-grained reduction to $\tilde{O}(\sqrt{n})$ \MF computations on $n$-node graphs.
\end{abstract}

\ifprocs

\fi 

\ifprocs
\maketitle
\fi

\ifprocs\else
\newpage
\fi

\section{Introduction}

A fundamental discovery of Gomory and Hu in 1961, now a staple of textbooks on Algorithms, says that every undirected graph can be compressed into a single tree while exactly preserving all pairwise minimum cuts (and maximum flows). 

\begin{theorem}[Gomory and Hu~\cite{GH61}]
Every undirected graph $G$ (even with edge weights)
has an edge-weighted tree $T$ on the same set of vertices $V(G)$ such that:
\begin{itemize} \compactify 
\item for all pairs $s,t\in V(G)$ the minimum $(s,t)$-cut in $T$ is also a minimum $(s,t)$-cut in $G$, and their values are the same.
\end{itemize}
Such a tree is called a \emph{cut-equivalent tree}, aka \emph{\GHT}. 
Moreover, the tree can be constructed in the time of $n-1$ queries to a \MF algorithm,%
\footnote{The notation \MF refers to the maximum $(s,t)$-flow problem, 
  which clearly has the same value as minimum $(s,t)$-cut.
  In fact, we often need algorithms that find an optimal cut (not only its value),
  which is clearly different (and usually harder) than \GMC.
}
where throughout $n=|V(G)|$.
\end{theorem}

In the interim sixty years, the \GHT has been investigated thoroughly from various angles (see Section~\ref{related_work}).
In spite of ingenious new approaches and much progress in simpler or harder settings, the chief open question 
remains elusive.

\begin{openq}
\label{oq1}
Can one construct a \GHT faster than the time of $O(n)$ \MF computations?
\end{openq}

In the most basic setting of \emph{simple graphs} (undirected, unweighted, no parallel edges), each \MF query can be answered in  $\tilde{O}(n^2)$ time%
\footnote{The notation $\tilde{O}(\cdot)$ hides a factor that is polylogarithmic in $n$.}
using the algorithm of Karger and Levine~\cite{KL15}, and the question becomes whether the natural cubic barrier can be broken.

\begin{openq}
\label{oq2}
Can one construct a \GHT of a simple graph $G$ in $o(n^3)$ time?
\end{openq}

Before the current work, subcubic%
\footnote{An algorithm is called subcubic if its running time is $o(n^3)$ where $n$ is the number of nodes. It is common to distinguish between mildly-subcubic $O(\frac{n^3}{\poly\log{n}})$ and truly-subcubic $O(n^{3-\eps})$. In this paper we refer to the latter.
}
algorithms were known only for sparse graphs \cite{BHKP07,KL15,AKT20}, planar graphs \cite{BSW15}, surface-embedded graphs \cite{BENW16}, and bounded-treewidth graphs \cite{ACZ98,AKT20_b}.

\subsection{Results}

We resolve Open Question~\ref{oq2} in the affirmative by giving the first subcubic algorithm for computing a \GHT for unweighted graphs.

\begin{theorem}
There is a randomized algorithm, with success probability $1-1/\poly(n)$, that 
constructs a \GHT of a simple graph $G$ in $\tilde{O}(n^{2.5})$ time.
\end{theorem}

Like the Gomory-Hu algorithm (and others), our new algorithm relies on queries to \MF on contracted graphs.
While the number of queries is still $\Omega(n)$, the gains come from reducing the sizes of the contracted graphs more massively and rapidly, making the total time proportional to only $\tilde{O}(\sqrt{n})$ queries.
The key ingredient is a new \expanders procedure that leads the algorithm to make \MF queries on ``the right'' pairs, guaranteeing that large portions of the graph get contracted. 
In fact, since these $\Omega(n)$ queries can be done in parallel inside each of $O(\log n)$ levels of recursion, they can be simulated using only $\tilde{O}(\sqrt{n})$ queries on graphs with $n$ nodes (by simply taking disjoint unions of the smaller graphs and connecting all sources/sinks into one source/sink, see e.g. \cite{LP20}). 
However, due to the contractions, these queries need to be performed on weighted graphs.

\paragraph{\APMF}
Once a \GHT is computed many problems become easy.
Notably, the first algorithm for computing the \EC of a graph, i.e. \GMC, computes the \GHT and returns the smallest edge \cite{GH61}. Faster near-linear time algorithms are known by now \cite{Gabow95,KS96,Karger00,KThorup19,henzinger2020local,ghaffari2020faster,MN20,GMW20}, but state-of-the-art algorithms for computing the maximum-flow value between all ${n \choose 2}$ pairs of nodes, called the \APMF problem,%
\footnote{Since we are in the setting of simple graphs, a possible name for the problem is \APEC. However, \EC is often used for the size of the \emph{global minimum cut} in simple graphs, not the $st$-version.
}
still rely on the reduction to a \GHT.
Indeed, to compute the maximum flow,
simply identify the lightest edge on the path between each pair in the tree,
which takes only $\tilde{O}(1)$ time per pair;
the bottleneck is computing the tree. 

The modern tools of fine-grained complexity have been surprisingly unhelpful in establishing a conditional lower bound for \APMF.
The Strong ETH\footnote{The Strong Exponential Time Hypothesis (SETH) of Impagliazzo and Paturi~\cite{ImpaSETH,CIP06} states that $k$-SAT cannot be solved in $(2-\eps)^n$ time for an $\eps>0$ independent of $k$. It has been a popular hardness assumption for proving polynomial lower bounds like $n^{3-o(1)}$ in recent years.} rules out fast algorithms in \emph{directed} graphs (where the \GHT does not exist \cite{HL07}) \cite{AVY18,KT18,A+18}, but it probably cannot give an $n^{2+\varepsilon}$ 
lower bound for simple graphs \cite{AKT20}.\footnote{Reducing $k$-SAT to \APMF would give a faster co-nondeterministic algorithm for $k$-SAT and refute the so-called Nondeterministic SETH \cite{carmosino2016nondeterministic}.}
Moreover, all natural attempts for placing it in the \APSP 
subcubic-hardness class \cite{VW18} had failed, despite the fact that \MF feels harder than \SPath, e.g., single-pair in $\tilde{O}(n^2)$ time is an easy coding-interview question for \SPath but certainly not for \MF.
It turns out that subcubic time is indeed possible, at least in unweighted graphs.

\begin{corollary}
There is a randomized algorithm, with success probability $1-1/\poly(n)$, that
solves \APMF in simple graphs in $\tilde{O}(n^{2.5})$ time.
\end{corollary}

For both \APMF and \APSP, cubic time is a natural barrier ---
shouldn't each answer take $\Omega(n)$ time on average? 
About thirty years ago, Seidel's algorithm \cite{seidel1995all} broke this barrier for \APSP in unweighted graphs using fast matrix multiplication with an $O(n^{\omega})$ upper bound; back then $\omega$ was $2.375477$ \cite{CW87} and today it is $2.37286$ \cite{AlmanW20}.
Our new $\tilde{O}(n^{2.5})$ algorithm for \APMF breaks the barrier using a different set of techniques, mainly expander decompositions, randomized hitting sets, and fast \MF algorithms (that in turn rely on other methods such as continuous optimization).
Interestingly, while designing an $O(n^{3-\eps})$ algorithm with ``combinatorial'' methods has been elusive for \APSP (see e.g. \cite{BansalW09,VW18,Chan12,Chan15,AbboudW14,ABV15,Yu18}), it can already be accomplished for \APMF. Our upper bound when using only combinatorial methods is $\tilde{O}(n^{2 \frac{5}{6}})$.

\paragraph{More Bounds}

Since the first version of this paper came out, a new algorithm has been published for solving a single \MF query in weighted graphs in time $\tilde{O}(m+n^{1.5})$ \cite{linearflow21}, improving over the previous bound $\tilde{O}(m\sqrt{n})$ \cite{LS14}, where throughout $m=|E(G)|$.
This development directly has led to the improvement of our bound for \GHT from $\tilde{O}(n^{2.75})$ in the previous version to $\tilde{O}(n^{2.5})$ in the current version.
In fact, our previous version had already included an $\tilde{O}(n^{2.5})$-time algorithm assuming a hypothetical $\tilde{O}(m)$-time \MF algorithm; it turns out that the additive term $+n^{1.5}$ in the new algorithm \cite{linearflow21} is inconsequential for our result.

The running time of our new algorithm improves beyond $n^{2.5}$ if the number of edges is below $n^2$. The precise bound is $\tilde{O}(n^{3/2}m^{1/2})$. However, for density below $n^{1.5}$, a previously known (and very different) algorithm \cite{AKT20} is faster; its time bound is $\tilde{O}(m^{3/2})$ under the hypothesis that \MF is solved in near-linear time, and using the new algorithm \cite{linearflow21} instead gives a slightly worse bound. 
The previous bounds and the state of the art for all density regimes are summarized in Figure~\ref{Bound:Existing:Algebraic:Theorem}.
It is also summarized in the following theorem together with the best combinatorial algorithms. In each item, the rightmost term is new, and the other terms (which are better for sparse graphs) are by previous work \cite{AKT20} (when plugging in \cite{linearflow21}).

\begin{theorem}
\label{mainthm}
There is a randomized algorithm, with success probability $1-1/\poly(n)$, that constructs a \GHT of a simple graph with $n$ nodes and $m$ edges in time:
\begin{enumerate}
\item\label{Bound:Existing:Algebraic:Theorem} $\tO\left(\min\set{m^{3/2}n^{1/6}, mn^{3/4}, \ n^{3/2}m^{1/2}} \right)$ using existing \MF algorithms,
\item\label{Bound:Existing:Combinatorial:Theorem}
$\tO\left( \min\set{m^{3/2}n^{1/3},\ n^{11/6}m^{1/2}} \right)$ using existing combinatorial \MF algorithms.
\end{enumerate}
\end{theorem}

\begin{figure}[ht]
  \begin{center}
    \includegraphics[width=3.3in]{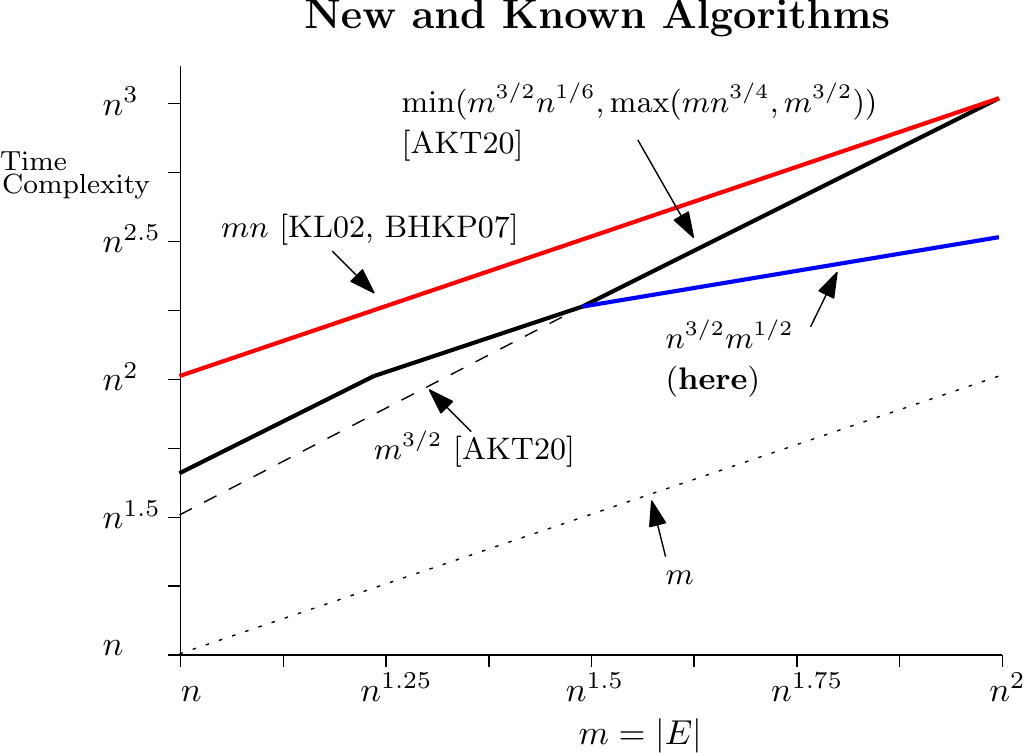}
  \end{center}
  \caption{State-of-the-art time bounds for constructing a \GHT of simple graphs.
    The dotted line is the best one can hope for (as no super-linear lower bound is known).
    The dashed line represents the $m^{3/2+o(1)}$-time algorithm~\cite{AKT20} that assumes a near-linear time \MF algorithm.
    }
    \label{Figs:Bounds}
\end{figure}

\subsection{Previous Algorithms}

Over the years, the time complexity of constructing a \GHT has decreased several times due to improvements in \MF algorithms, but there have also been a few conceptually new algorithms. 
Gusfield~\cite{Gusfield90} presented a modification of the Gomory-Hu algorithm in which all the $n-1$ queries
are made on the original graph $G$ (instead of on contracted graphs). 
Bhalgat, Hariharan, Kavitha, and Panigrahi~\cite{BHKP07} designed an  $\tilde{O}(mn)$-time algorithm utilizing a tree packing approach~\cite{Gabow95,Edmonds70} that has also been used in other algorithms for cut-equivalent trees~\cite{Cole03,HKP07,AKT20}.
In particular, they designed an $O(mk)$-time algorithm for constructing a $k$-partial \GHT, which preserves the minimum cuts if their size is up to $k$ (see \cite{Panigrahi16} and the full version \cite{BCHKP08}).
A simple high-degree/low-degree strategy is helpful in sparse graphs:
only $\sqrt{m}$ nodes can have degree (and therefore outgoing flow) above $\sqrt{m}$, thus a $\sqrt{m}$-partial tree plus $\sqrt{m}$ \MF queries are sufficient,
which takes $O(m^{3/2})$ time if \MF is solved in linear time.
Using the current \MF algorithms, this strategy results in the time bound $\tO(\min\set{m^{3/2} n^{1/6}, \max\set{mn^{3/4},m^{3/2}}})$ \cite{AKT20}. 
Additionally, two recent algorithms accelerate the Gomory-Hu algorithm by making multiple steps at once to achieve $\tilde{O}(m)$ time, one requires nondeterminism \cite{AKT20} and the other requires a (currently non-existent) fast minimum-cut data structure \cite{AKT20_b}.
The new algorithm uses a similar approach.

In weighted graphs, the bound for \GHT is still $n$ times \MF and therefore cubic using the very recent $\tilde{O}(m+n^{1.5})$ algorithm for \MF \cite{linearflow21}.
If the graph is sparse enough and the largest weight $U$ is small,
one can use algorithms~\cite{madry2016computing,LS19,LS20FOCS}
that run in time $\tilde{O}( \min\{m^{10/7}U^{1/7}, m^{11/8}U^{1/4},$ $m^{4/3}U^{1/3}\} )$
that might give a better time bound.
It is likely that the new techniques will lead to faster algorithms for this case as well, although new ideas are required.

Further discussion of previous algorithms can be found at the Encyclopedia of Algorithms~\cite{Panigrahi16}.

\subsection{Technical Overview}
\label{overview}

This section attempts to explain the new algorithm in its entirety in a coherent and concise way.
This is challenging because the algorithm is complicated both conceptually, relying on several different tools, and technically, involving further complications in each ingredient.
The main text gives the overall algorithm while making nonobvious simplifying assumptions, and meanwhile the text inside the boxes discusses how these assumptions are lifted or how to handle other technical issues that arise.
The reader is advised to skip the boxes in the first read.

\medskip

Given a graph $G=(V,E)$ the goal is to compute its cut-equivalent tree $T$.
The new algorithm, like the Gomory-Hu algorithm builds $T$ recursively, by maintaining an intermediate (partial) tree $T'$ that gets \emph{refined} with each recursive call.
The nodes of $T'$ are subsets of $V$ and are called super-nodes.
Each recursion step considers a super-node of $T'$ and splits-off parts of it.
In the beginning, there is a single super-node and $T'=\{V\}$.
If the minimum $(s,t)$-cut $(S,V\setminus S)$ for a pair $s,t\in V$ is found, then $T'$ is refined by splitting $V$ into two super-nodes $S$ and $V \setminus S$ and by connecting them with an edge of weight $\lambda_{s,t}$ the connectivity between $s$ and $t$.
Then, the main observation is that when refining $S$ the entire super-node $V\setminus S$ can be contracted into a single node; this does not distort the connectivities for pairs in $S$ and it ensures that the refinements of $S$ and of $V \setminus S$ are consistent with one another (and can therefore be combined into the final tree $T$).
More generally, consider an intermediate tree $T'$ whose super-nodes $V_1,\ldots,V_l$ form a partition $V=V_1\sqcup\cdots\sqcup V_l$. 
The algorithm then refines each super-node $V_i$ by operating on an \emph{auxiliary graph} $G_i$ that is obtained from $G$ by contracting nodes outside $V_i$ in a manner informed by $T'$: each connected component of $T'$ after the removal of $V_i$ is contracted into a single node.

Now, consider such super-node $V_i$ in an intermediate tree $T'$ along with its auxiliary graph $G_i$,%
\footnote{For most of this overview, it is safe to think of the very first iteration where $V_i=V$ and $G_i=G$. Later on we will point out the complications that arise due to the existence of contracted nodes $V(G_i) \setminus V_i$.}
and let $T_i$ be its (unknown) cut-equivalent tree,
i.e., the subtree of $T$ induced on $V_i$.
The goal is to refine $V_i$.
Gomory and Hu took an arbitrary pair $s,t \in V_i$ and used an $s,t$-\MF query to get a cut. (For a more detailed exposition of the Gomory-Hu algorithm, see Section~\ref{prelim:GH} and Figure~\ref{Figs:GH_trees_3}).
Another natural approach is to take the global minimum cut of $V_i$.
Either way, the depth of the recursion could be $n-1$ if each refinement only splits-off one node from $V_i$.%
\footnote{Another natural approach is to look for a \emph{balanced} minimum cut. It is not clear that such a cut exists or that it can be computed efficiently. However, it is fair to say that part of the new algorithm uses a similar approach.}
The new algorithm aims to finish within recursion depth $O(\log{n})$, and the strategy is to refine $V_i$ into multiple disjoint minimum cuts $V = V_{i,1} \sqcup \cdots \sqcup V_{i,k}$ at once, where $|V_{i,j}| \leq |V_i|/2$ ensuring that the depth of the recursion is logarithmic.

The idea is to pick a pivot node $p \in V_i$ (can be thought of as a root for $T_i$) and to compute \MF between $p$ and every node in $V_i$.
That is, for each node $v \in V_i$, the algorithm computes a minimum $(p,v)$-cut $(C_v,V_i\setminus C_v)$ where $v \in C_v$ and $p \in V \setminus C_v$.
While this gives a lot of information, it is not necessarily sufficient for computing $T_i$ because these cuts may not determine an optimal cut for other pairs in $V_i$.%
\footnote{An extreme scenario is if for every $v \in V_i$
  the minimum $(p,v)$-cut is $(\{p\},V_i \setminus \{p\})$; 
  clearly, this information is not sufficient for knowing the whole tree.
  But this will not happen if we choose a good pivot $p$. 
}
Still, it allows us to make progress by splitting-off cuts $V_{i,j}$ from $V_i$ with $|V_{i,j}|\leq |V_i|/2$ such that $V_{i,j}=C_v$ for some $v\in V_i$ and can therefore be safely used to refine $V_i$.
This approach indeed works if $p$ happens to be a \emph{good} pivot,
which means that it is centroid-like in $T_i$ in the sense that most cuts $C_v$ have less than half the nodes.
Moreover, for the correctness of the algorithm, a good pivot must satisfy that the ``depth'' of $T_i$ from $p$ is $O(\sqrt{n})$.
It turns out that both requirements can be accomplished.

\defproblem{Complication 1: How do we get a good pivot?}
{
A recent nondeterministic algorithm \cite{AKT20} chooses a pivot by guessing the centroid (there are other, more crucial uses for guessing in that algorithm).
An NC algorithm of Anari and Vazirani \cite{AnariV20} for planar graphs tries all nodes as pivots in parallel.
  Bhalgat et al. \cite{BHKP07} and another recent work \cite{AKT20_b} choose a pivot at random in each iteration and argue that the recursion depth is bounded by $O(\log n)$ with high probability.
The new algorithm takes a similar approach but requires more care, because now a bad pivot not only slows down progress but it could lead to wrong cuts.
For the correctness of what follows, a good pivot must also satisfy that a certain notion of depth of $T_i$ when rooted at $p$ is bounded by $O(\sqrt{n})$.
And it must hold, with high probability, for every single iteration.
Thus, to decrease the failure probability from constant to $1/\poly(n)$, the algorithm chooses a random set $S$ of $\tilde{O}(\sqrt{n})$ candidate pivots,
computes a refinement of $T'$ by following the standard Gomory-Hu algorithm while only picking pairs from $S$ using $|S|-1$ queries to \MF,
and then picks the pivot $p$ to be the node in $S$ whose component in the refinement is largest.
It follows that $p$ is both centroid-like ($|C_v| \leq |V_i|/2$ for most $v\in V_i$) and its component (after this refinement) has a small depth. 
Notably, derandomizing this part would lead to a deterministic algorithm with $n^{1-\eps}$ queries to \MF.
}

At this point, the algorithm has a good pivot $p$ for $G_i$ and the goal is to compute the minimum $(p,v)$-cut for all nodes $v \in V_i$.
To focus this exposition on the novel aspects of the new algorithm, let us make a few simplifying assumptions.
First, assume that there is a unique \GHT $T$ for $G$ (and therefore there is a unique $T_i$ for $G_i$).

\defproblem{Complication 2: Operating as if there is a unique tree.}
{
A standard way to make the \GHT unique is to perturb edge weights in $G$ by small $1/\poly(m)$ additive terms to ensure that all minimum cuts are unique.
However, this would prevent the algorithm from using known \MF algorithm for unweighted graphs and hurt the running time, so it cannot  be done in our real algorithm.
The uniqueness of the tree is important both for the analysis, as one can safely talk about \emph{the} minimum cut $C_v$,
and for the correctness --
what guarantees that cuts $C_v$ for different $v \in V_i$ are all consistent with a single \GHT?
Our real algorithm escapes these issues by working with \emph{latest} minimum cuts, i.e., a minimum $(p,v)$-cut $(C_v,V\setminus C_v)$ that has the smallest possible $|C_v|$. 
A similar approach was used by Bhalgat et al. \cite{BHKP07} and others \cite{BGK20esa}. See Section~\ref{sec:latest} for more background on latest cuts.
}

The second simplifying assumption is that the weights in $T_i$ are decreasing along any path from the root $p$.
As a result, the minimum $(p,v)$-cut for every node $v\in V_i$ is exactly its subtree; this greatly simplifies describing and analyzing the cuts with respect to the tree (see Figure~\ref{Figs:Expanders}).

\defproblem{Complication 3: Cut-membership trees.}
{
In general, the lightest edge on the path from a node $v\in V_i$ to $p$ in $T_i$ is not necessarily the first edge (right ``above'' $v$). 
In other words, the minimum $(p,v)$-cut $C_v$, whose nodes we call \textit{cut-members} of $v$, could be a strict superset of $v$'s subtree in $T_i$,
which complicates the analysis of subsequent ingredients, as they rely on bounding the number of cut-members of nodes. 
A convenient tool for reasoning about this is the cut-membership tree $\TG_p$ with respect to $p$ \cite[Section $3$]{AKT20_b}. This is a coarsening of $T_i$ where all nodes $B \subseteq V_i$ whose minimum cut to $p$ is the same are merged into one bag
\ifprocs
.
\else
(see Figure~\ref{Figs:Tzeta}).
\fi
Importantly, it is still a tree and it satisfies the assumption of decreasing weights, meaning that the minimum cuts are indeed always subtrees.
All of the analysis is carried on with $\TG_p$ rather than $T_i$.
  }

The algorithm now creates for each node $v \in V_i$ 
an estimate $c'(v)$ for the connectivity $\lambda_{p,v}$ between $p$ and $v$;
it is initialized to $c'(v)=\deg(v)$, which is always an upper bound.
Then, the algorithm repeats procedure \expanders described below $O(\log{n})$ times; each iteration finds new $(p,v)$-cuts and updates these estimates by  keeping for each $v$ the minimum value seen so far. 
Since each cut is an upper bound on the connectivity, the estimates $c'(v)$ never decrease below $\lambda_{p,v}$, and with high probability they are eventually tight for all nodes.
A node $v$ is called \emph{done} if $c'(v)=\lambda_{p,v}$ and \emph{undone} otherwise.

In the beginning a node is done if and only if it is a leaf in $T_i$. (Imagine a tree with green leaves and red internal nodes.)
Each iteration of the \expanders procedure is makes a node done if its subtree contains at most $\sqrt{n}$ undone nodes, with high probability. (At each step, a red node becomes green if its subtree contains at most $\sqrt{n}$ red nodes.)
How many iterations until all nodes are done (green)?
Using the fact that $p$ is a good pivot and the depth of $T_i$ is $O(\sqrt{n})$, it can be shown that with high probability $O(\log{n})$ iterations suffice.

\paragraph{ Procedure \expanders}
For each value $w=2^{i'}$ for $i'=0,1,\dots,\log n$ there is a sub-procedure
that aims to find $C_v$ for all nodes $v \in V_i$ such that $w \leq \lambda_{p,v}<2w$. 
Fix such a node $v$ that is undone; then $\deg(v)\geq c'(v)>\lambda_{p,v}\geq w$.
Thus, only nodes of degree $>w$ are targeted by the sub-procedure so we call them \emph{relevant}.

In a preprocessing phase, for each $w=2^{i'}$,
the algorithm prepares an expander decomposition $V = H^w_1 \sqcup \cdots \sqcup H^w_h$ of the entire graph $G$ with expansion parameter $\phi_w=1/\sqrt{w}$.
Each component $H^w_j$ is a $\phi_w$-expander, meaning that there is no sparse cut inside it, and the total number of edges outside the expanders is bounded $\tfrac12 \sum_j \delta(H^w_j) = O(|E(G)|\cdot \phi_w\cdot \log^3{n})$.
Efficient and simple algorithms for computing this decomposition exist, e.g.~\cite{SW19}.

\defproblem{Complication 4: Lower bounding $w$.}
{
The expander decomposition algorithm requires the parameter $\phi$ to be $O(1/\log{n})$, and furthermore it must be $O(1/\log^3 n)$ for the outside edges bound to be meaningful. 

When $\phi_w=1/\sqrt{w}$ this is only an issue for very small $w$ (and can be resolved by decreasing $\phi_w$ by a $\log^3 n$ factor).
Our combinatorial $\tilde{O}(n^{2\frac{5}{6}})$ time algorithm, however, uses $\phi_w = n^{1/3}/w$, thus $w$ must be at least $\Omega(n^{1/3}\log^3 n)$.
For this reason, the very first ``refinement'' step of the algorithm is to compute a $k$-partial tree with $k=O(n^{1/3}\log^3 n)$ using the $O(nk^2)$-time algorithm of Bhalgat et al. \cite{BHKP07} (based on tree packings).
This gives a partial tree $T'$ in which two nodes $u,v$ are in separate super-nodes if and only if their connectivity is at most $k$.
Afterwards, values $w<n^{1/3}\log^3 n$ can be ignored as pairs in the same super-node $u,v \in V_i$ are guaranteed to have $\lambda_{u,v}>k$.
  }

Let $H_v$ be the expander containing $v$; it could be of one of three kinds, each kind is solved by a different method (even though neither $v$ nor $H_v$ is known to the algorithm). Let $L = C_v \cap H_v$ be the piece of $H_v$ that falls inside $v$'s subtree $C_v$ and let $R = (V \setminus C_v) \cap H_v$ be the remainder of $H_v$; we call them the left and right parts of $H$, respectively, see Figure~\ref{Figs:Expanders}.

\begin{enumerate}
\item\label{case1} Small: $H_v$ contains $|H_v|< w/8$ nodes.
\item\label{case2} Large righty: $|H_v|\geq w/8$ but contains only $|L| \leq 2\sqrt{w}$ relevant nodes from $C_v$. 
\item\label{case3} Large lefty: $|H_v|\geq w/8$ but contains only $|R| \leq 2\sqrt{w}$ relevant nodes from $V \setminus C_v$. 
\end{enumerate}

A key observation is that $H_v$ cannot have both $|L|,|R| > 2\sqrt{w}$ due to the fact that $H_v$ is a $1/\sqrt{w}$-expander: the cut $(L,R)$ in $H_v$ has at most $\lambda_{p,v}<2w$ edges and thus its volume $\min (\vol (L), \vol(R))\geq \min (|L| w, |R|w)$ cannot be more than $2w\sqrt{w}$. 
Therefore, a large expander must either be lefty or righty; it cannot be balanced.
This analysis assumes that all nodes in $L$ and $R$ are relevant, but in general they could contain many low-degree nodes.
Handling this gap is perhaps the most interesting complication of this paper, distinguishing it from recent applications of expander-based methods for \GMC where low-degree nodes are not an issue: if the global minimum cut has value $w$ then all nodes must have degree $\geq w$.

\defproblem{Complication 5: Low-degree nodes in expanders.}
{
Our use of expander-decompositions is nonstandard but is reminiscent of the way they are used in two recent algorithms for \GMC (which is equivalent to finding the smallest edge in the Gomory-Hu Tree).
The Kawarabayashi-Thorup~\cite{KThorup19} technique (as described by Saranurak~\cite{saranurak2020simple}) takes each expander, \emph{trims} and \emph{shaves} it to make sure it is exclusively on one side of the cut, and then contracts it. 
It is shown that most edges remain inside the expanders (and are therefore contracted away) despite the trimmings.
The algorithm of Li and Panigrahy~\cite{LP20} maintains a set of candidates $U$ containing nodes on each side of the cut and then uses the expanders to iteratively reduce its size by half. 
Some nodes of each expander are kept in $U$, and it is argued that there must be expanders that are mostly on the left and expanders that are mostly on the right.

Unfortunately, neither of these approaches seems to work when searching for all minimum $(p,v)$-cuts rather than for a single global minimum cut.
At their bottom, both rely on the following observation: if $L$, the left side of $H$, contains $\geq \ell$ nodes that are committed to $H$, i.e., most of their edges stay inside $H$, then the volume of $L$ is $\Omega(\ell w)$.
This is because when the global minimum cut has value $w$, all nodes must have degree (or capacitated degree in the weighted case) at least $w$.
However in our setting, the minimum $(p,v)$-cut $C_v$ does not have to be minimal for any $u \in C_v$ except $v$, and nodes on the left could have arbitrarily small degrees (see Section~\ref{sec:clcr} for an extreme example).

Consequently, the arguments in this paper are a little different and only lead to $O(\sqrt{w})$ savings rather than $\Omega(w)$.
If there was a magical way to get rid of nodes of degree $<w$, then a near-linear time algorithm could follow.
The $k$-partial tree (from Complication 4) indeed only leaves high-degree nodes in $V_i$ but the contracted nodes $V(G_i)\setminus V_i$ could have arbitrary degrees.
Instead, the algorithm and the analysis reason about the subsets  $\hat{H},\hat{L},\hat{R}$ of $H,L,R$ containing high degree nodes, as well as $H,L,R$ themselves.
With some care, the arguments presented in this section can indeed be made to work despite the gap.
 }

\begin{figure}[ht]
  \begin{center}
    \includegraphics[width=4.3in]{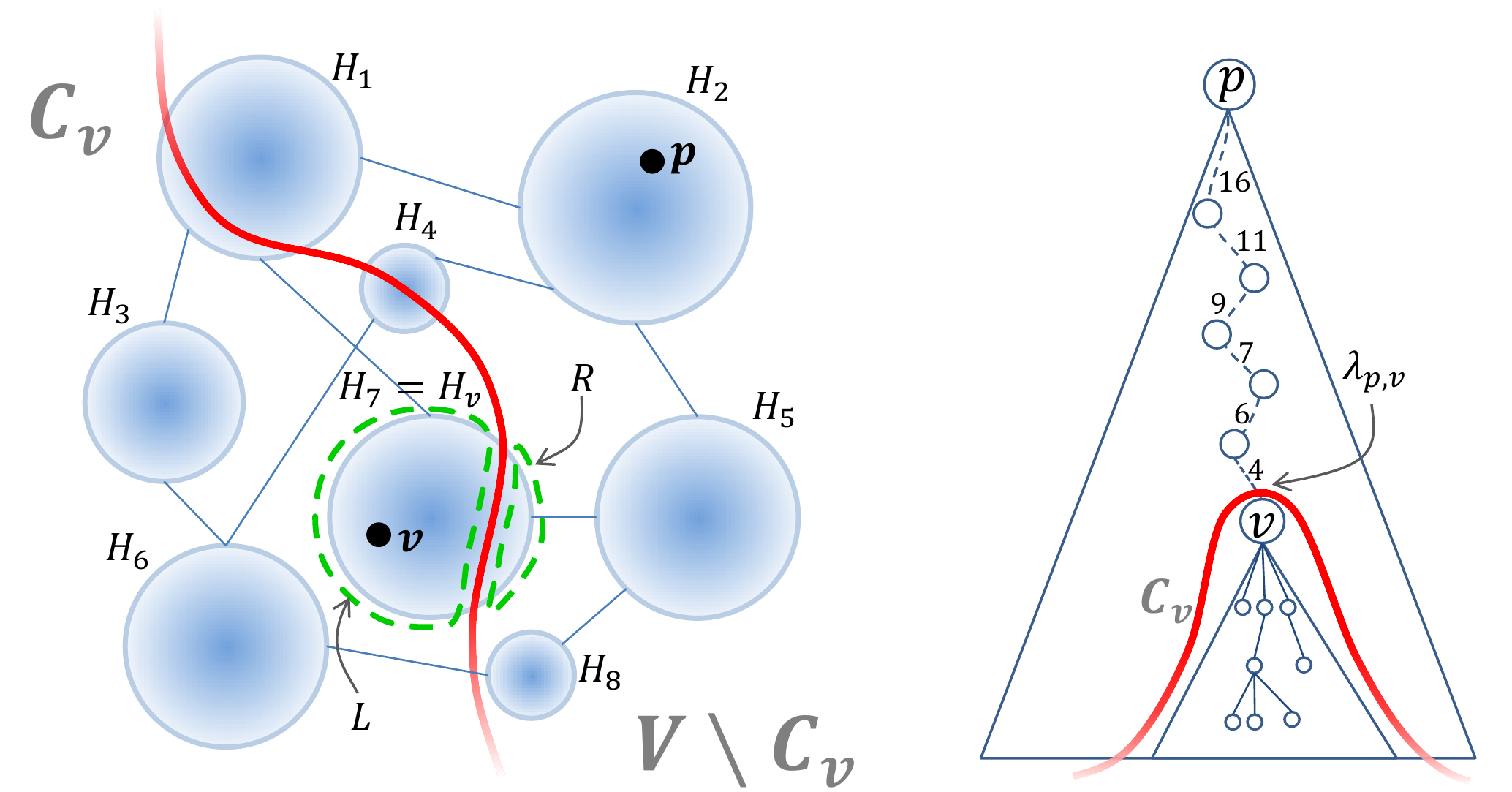}
  \end{center}
  \caption{
    An expander decomposition of the graph $G$ (on left), and the (unknown) cut-equivalent tree $T$ rooted at pivot $p$ (on right).
    The unique minimum $(p,v)$-cut $(C_v, V \setminus C_v)$ for an arbitrary node $v$ is depicted by a red curve in both figures.
    The algorithm must find $C_v$. 
    The edges in $T$ are assumed to be decreasing, thus $C_v$ is exactly $v$'s subtree with an edge of weight $\lambda_{p,v}$ above it.
    The expander decomposition of $G$ is $\{H_i\}_{i=1}^8$,
    and the expander containing $v$ is $H_v=H_7$.
    The sets $L$ and $R$ are depicted by dashed green lines;
    at least one of them must be small.
    The three kinds of expanders can be seen: small ($H_4,H_8$), large lefty ($H_7,H_3,H_6$), and large righty ($H_1,H_2,H_5$).
    Since $H_v$ is a large lefty expander, the cut $C_v$ should be found when handling \text{\bf Case~\ref{case3}}. 
    }
    \label{Figs:Expanders}
\end{figure}

To handle \text{\bf Case~\ref{case1}} the algorithm simply asks a \MF query between $p$ and each \emph{relevant} node $u$ in each small expander $H$, i.e., $|H|<w/8$.
Clearly, if $v$ happens to be in a small expander it gets queried and its optimal cut is found with probability $1$.
The slightly trickier part is arguing that only $t=O(\sqrt{n})$ queries are performed: any queried node $u$ has degree $\geq w$ (since it is relevant) while a small expander only has $<w/8$ nodes, thus overall $\Omega(tw)$ edges are outside the expanders.
But the expander decomposition guarantees that only $O(|E|/\sqrt{w})$ edges are outside the expanders, and thus $t = O(|E|/w^{1.5})$; in the hardest case where $w=\Omega(n)$ this is $t=O(\sqrt{n})$.
The fact that the auxiliary graph $G_i$ is a multigraph does not matter, because the expander decomposition is for $G$.
This is the only case in the algorithm that requires $G$ to be simple;%
\footnote{Besides, of course, to speed up the \MF queries: since all $\lambda_{s,t}\leq n$, the time bound of Karger and Levine~\cite{KL15} is $\tilde{O}(n^2)$.
}
the argument is similar to \cite[Observation 5]{KThorup19}.

\defproblem{Complication 6: Saving when $w=o(n)$ using Nagamochi-Ibaraki sparsification.}
{
In the hardest case where $w=\Omega(n)$ the number of queries $t=O(|E|/w^{1.5})$ in Case~\ref{case1} is already $O(\sqrt{n})$, but to get the $n^{2.5+o(1)}$ bound in general, the algorithm utilizes two sparsifications.
First, it computes the expander-decomposition for each $w=2^i$ on a sparsifier $G_w$ with $|E_w|=O(nw)$ rather than on $G$.
The sparsifier of Nagamochi and Ibaraki~\cite{NI92} with parameter $k$ has $O(nk)$ edges and ensures that all cuts of value up to $k$ are preserved.
This gives a better upper bound on the number of edges outside expanders and leads to $t=O(|E_w|/w^{1.5})=O(n/\sqrt{w})$.
Second, while the number of queries exceeds $\sqrt{n}$ when $w=o(n)$, the algorithm saves by computing each of these queries in $\tilde{O}(nw)$ time rather than $O(m)$ by operating on a sparsifier of the same kind, but now for the auxiliary graph $G_i$, not $G$.
This does not introduce error as the algorithm is only interested in cuts of value $\leq 2w$ and the total time becomes $\tilde{O}(n^2\sqrt{w}) = \tilde{O}(n^{2.5})$.
  }

Handling \text{\bf Case~\ref{case2}} involves a surprising \pC procedure that can \emph{almost} compute single-source all-sinks \MF with $O(\log n)$ queries to \MF.
It is a tricky but lightweight reduction that only involves contracting subsets of nodes.

\begin{lemma}[\cite{LP20}, see also Lemma~\ref{thm:proc_main}]\footnote{In the late stages of writing this paper, we have encountered the very recent work of Li and Panigrahy (FOCS 2020) and found out that they have also discovered this simple yet powerful procedure, naming it \emph{Isolating Cuts Lemma} \cite[Theorem II.2]{LP20}. We have kept our proof in 
\ifprocs
the full version,
\else
Section~\ref{app:proc} (see also Section~\ref{sec:proc}) 
\fi
as it has a different flavor. The usage is similar but different: there, it deterministically finds the global minimum cut when it is unbalanced (has a small side), while here it finds many small subtrees at once.
}
Given a pivot $p$ and a set of terminals $C$, procedure \pC uses $O(\log{n})$ \MF queries and returns an estimate $(p,v)$-cut $S_v$ for each terminal $v\in C$ such that:
for all $v \in C$, if the minimum $(p,v)$-cut $C_v$ satisfies $C_v \cap C = \emptyset$ (i.e., this cut isolates terminal $v$) then $S_v=C_v$.
\end{lemma}

This procedure is quite strong in the sense that it solves hard cases with $\tilde{O}(1)$ \MF queries when other approaches require $\Omega(n)$ queries, for instance in cases where all subtrees are small (see Section~\ref{sec:proc}).
But it is also weak in the sense that there is no way to distinguish the correct answers from fake ones. Nodes with large subtrees may never be isolated yet the procedure returns a (fake) small cut that, if used by the algorithm, could lead to a wrong tree.
The canonical such example is discussed in Section~\ref{sec:proc} as it motivates the use of an expanders-based approach.

The \pC procedure is handy in \text{\bf Case~\ref{case2}} where the algorithm takes each expander $H$ with $|H|\geq w/8$ (there are only $O(n/w)=O(1)$ such expanders) and guesses that $H=H_v$ and that it is a large righty expander, i.e., $|L|\leq 2\sqrt{w}$.
In this case, if indeed $v \in H$, it is possible to pick $t=\tilde{O}(\sqrt{w})$ sets of terminals $C_1,\ldots,C_t$ such that at least one of them isolates $v$:
for each $i'\in [t]$ include each node $u \in H$ into $C_{i'}$ with probability $1/\sqrt{w}$.
For a set $C_{i'}$ to isolate $v$, it must (1) pick $v$ and (2) not pick any other node from $L$, which happens with probability $\Omega(1/\sqrt{w})$; therefore, with high probability at least one of the sets $C_1,\ldots,C_t$ isolates $v$.
The \pC procedure is called for each $C_{i'}$, getting new estimates for all $v \in V_i$; a total of $\tilde{O}(\sqrt{w})$ queries for all large expanders.
The intuition is that a large righty expander helps the algorithm pick a good set of terminals that isolates $v$, a highly nontrivial task if $v$'s subtree is large, by focusing the attention on a component that (presumably) only contains few nodes from the subtree.
Indeed, the \pC procedure is not as helpful in \text{\bf Case~\ref{case3}} where $H_v$ is a large lefty expander and $|L|$ can be up to $n$, since randomly chosen terminals from $H$ are unlikely to leave $v$ isolated.

To handle the final \text{\bf Case~\ref{case3}}, the algorithm identifies the set $\text{Top}_{k}(H)$ of the top $k=3\sqrt{n}+1$ nodes in terms of their $c'(u)$ estimate for each large expander $H$, and performs a \MF query for each one. 
Since there are only $O(n/w)=O(1)$ large expanders, only $O(\sqrt{n})$ queries are performed.
Suppose that $v$ is indeed in a large lefty expander $H_v$ with $|R|\leq 2\sqrt{w}$.
The argument is that $v$ must be among the queried nodes unless there are $>\sqrt{n}$ undone nodes $u \in C_v$ in its subtree.%
\footnote{Recall that a node $u$ is undone if $c'(u)> \lambda_{p,u}$ and done if $c'(u)= \lambda_{p,u}$.
}
This is because any done node $u \in C_v$ has $c'(u) = \lambda_{p,u} \leq \lambda_{p,v}$, since $C_v$ is a valid $(p,u)$-cut of value $ \lambda_{p,v}$, while $\lambda_{p,v} < c'(v)$ (otherwise $v$ is already done).
Therefore, any node $u \in H_v$ with $c'(u) \geq c'(v)$ can either be in $R \subseteq V \setminus C_v$ or it could be an undone node in $L \subseteq C_v$.
There are at most $2\sqrt{w}+\sqrt{n}$ such nodes in total; thus $v \in \text{Top}_{k}(H_v)$.
Observe that this argument does not work in \text{\bf Case~\ref{case2}} where $H_v$ is a large righty expander and $|R|$ can be up to $n$, since the nodes in $R$ might have larger connectivity than $v$.

Finally, if $v$ happens to be in a large lefty expander when its subtree contains $>\sqrt{n}$ undone nodes, then the algorithm is not guaranteed to find $C_v$ in this iteration of procedure \expanders.
($v$ is a red node with many red nodes in its subtree.)
Soon enough, within $O(\log n)$ iterations, the algorithm gets to a point where most of $v$'s subtree has become done and then $C_v$ is found.
(The green from the leaves quickly ``infects'' the entire tree.)

\subsection{Related Work}
\label{related_work}

\paragraph{Harder settings}
On the hardness side, the only related lower bounds are for \APMF in the harder settings of directed graphs \cite{AVY18,KT18,A+18} or undirected graphs with node weights \cite{AKT20},
where Gomory-Hu trees cannot even exist, 
because the $\Omega(n^2)$ minimum cuts might all be different~\cite{HL07} (see therein also an interesting exposition of certain false claims made earlier).
In particular, SETH gives an $n^{3-o(1)}$ lower bound for weighted sparse directed graphs \cite{KT18} and the $4$-Clique conjecture gives an $n^{\omega+1-o(1)}$ lower bound for unweighted dense directed graphs \cite{A+18}.
Nontrivial algorithms are known for unweighted directed graphs, with time $O(m^{\omega})$ \cite{CLL13} (fast matrix multiplication techniques have only been helpful in directed graphs so far), and also for special graph classes such as planar \cite{LNSW12} and bounded-treewidth \cite{ACZ98,AKT20_b}.
Moreover, algorithms exist for the case we only care about pairs of nodes whose Max-Flow value is bounded by small $k$~\cite{BHKP07,2ECB,A+18}.
 Generalizations of the \GHT to other cut requirements such as multiway cuts or cuts between groups of nodes have also been studied \cite{Hassin88,Hassin90,Hassin91,Har01,EH05,CKK16}.

\paragraph{Approximations}
Coming up with faster constructions of a \GHT at the cost of approximations has also been of interest, see e.g.~\cite{Panigrahi16}, with only few successful approaches so far.
One approach is to sparsify the graph into $m'=\tO(\eps^{-2} n)$ edges
in randomized $\tO(m)$ time using the algorithm of Benczur and Karger~\cite{BK15} (or its generalizations), and then apply an exact \GHT algorithm on the sparse (but weighted) graph. 
Unfortunately, even when aiming for an approximate tree, each query throughout the Gomory-Hu (or the new) algorithm must be exact (see \cite{Panigrahi16,AKT20_b}).
Therefore, with current \MF algorithms, a $(1+\eps)$-approximate \GHT of unweighted graphs can be constructed in $\tO(\eps^{-2} n^{2.5})$ time.
Using a different approach that produces a flow-equivalent tree (rather than a cut-equivalent tree, meaning that the minimum cuts in the tree, viewed as node bipartitions, might not correspond to minimum cuts in the graph), one can design a $(1+\eps)$-\APMF algorithm that runs in time $\tilde{O}(n^2)$ \cite{AKT20_b}.
Finally, one can use \Raecke's approach to compute a cut-sparsifier tree \cite{Rac02}, which has a stronger requirement (it approximates all cuts of $G$) but can only give polylogarithmic approximation factors.
Its fastest version runs in almost-linear time $m^{1+o(1)}$ and achieves $O(\log^{4}n)$-approximation \cite{RST14}.

\paragraph{Applications and experimental studies}
Cut-equivalent trees have appeared in countless application domains.
One example is the pioneering work of Wu and Leahy~\cite{WL93} in 1993 on image segmentation using \GHT that has evolved into the \emph{graph cuts} paradigm in computer vision.
Another example is in telecommunications where Hu \cite{Hu74} showed that the \GHT is the optimal solution to the \emph{minimum communication spanning tree} problem; consequently there is interest in characterizing which graphs have a \GHT that is a \emph{subgraph} \cite{korte2012combinatorial,Naves18}.
In mathematical optimization, a seminal paper of Padberg and Rao \cite{PR82} uses the \GHT to find odd cuts that are useful for the $b$-matching problem (and that have been used in a breakthrough NC algorithm for perfect matching in planar graphs \cite{AnariV20}).
The question of how the \GHT changes with the graph has arisen in applications such as energy and finance and has been investigated, e.g. \cite{picard1980structure,barth2006revisiting,hartmann2013dynamic}, starting with Elmaghraby in 1964 \cite{elmaghraby1964sensitivity} and up until very recently \cite{BGK20esa}.
Motivated by the need of a scalable algorithm for Gomory-Hu Tree, Akiba et al. \cite{akiba2016cut} have recently introduced a few heuristic ideas for getting a subcubic complexity in social networks and web graphs. 
Earlier, Goldberg and Tsioutsiouliklis \cite{GT01} conducted an experimental comparison of the Gomory-Hu and Gusfield's algorithms.

\paragraph{Expander Decompositions}
A key ingredient of our \GHT algorithm is an expander-decomposition of the graph \cite{KVV04,OV11,OSV12,ST13,SW19,chuzhoy2019deterministic}.
Such decompositions have led to several breakthroughs in algorithms for basic problems in the past decade, e.g. \cite{ST14,KLOS14,NSW17}.
A typical application solves the problem on each expander and then somehow combines the answers, treating each expander as a node. The clique-like nature of each expander and the sparsity of the outer graph lead to gains in efficiency.
This approach does not seem helpful for \GHT since there may not be any connection between the tree and the decomposition.
The application in this paper is more reminiscent of recent deterministic \GMC algorithms \cite{KThorup19,saranurak2020simple,LP20}, but is also different from those (see Complication 5 in Section~\ref{overview}).

\section{Preliminaries}

\subsection{General Notations}
We will mostly work with unweighted graphs $G=(V,E)$, but throughout our algorithms we might contract vertices and end up with auxiliary graphs $G=(V,E,c)$ that are weighted $c:E \to [U]$, i.e. with capacities in $[U]=\{1,\ldots,U\}$ on the edges.
All graphs in this paper will be undirected.
We denote by $\deg(v)$ and $\cdeg(v)$ the number of edges and the total capacity on edges incident to $v\in V$, respectively. 
We treat \emph{cuts} as subsets $S\subset V$, or partitions $(S, V\setminus S)$.
The \emph{value} of a cut $S$ is defined as $\delta(S)=\card{\{\{u,v\}\in E : u\in S,v\in V\setminus S\}}$, and if two subsets $S,T\subseteq V$ are given, then $\delta(S,T)=\card{\{\{u,v\}\in E : u\in S,v\in T\}}$.
When the graph is weighted we define the values as $\delta(S)=\sum_{\{u,v\}\in E, u\in S,v\in V\setminus S} c(u,v)$ and $\delta(S,T)=\sum_{\{u,v\}\in E, u\in S,v\in T} c(u,v)$. 
\ifprocs
For a pair $u,v\in V(G)$, we denote by $\lambda_{u,v}$ or $\MF(u,v)$ the value of a minimum $(u,v)$-cut, $\delta(S)$.
\else
For a pair $u,v\in V(G)$, we denote by $\MC(u,v)$ a minimum cut $(S,V\setminus S)$ between $u$ and $v$, and by $\lambda_{u,v}$ or $\MF(u,v)$ the value of this cut, $\delta(S)$.
\fi

\subsection{Max Flow Algorithms: Unweighted, Weighted, and Combinatorial}
\label{sec:comb}

We use as a black box known algorithms for \MF to get a minimum $(u,v)$-cut for a given pair $u,v$.
Throughout the paper, three existing algorithms are used.
First, if the flow size is bounded, the Karger-Levine algorithm~\cite{KL15} that runs in time $\tO(m+nF)$ where $F$ is the size of the maximum flow is particularly fast.
Second, for larger flows we use the very recent algorithm~\cite{linearflow21}  that runs in time $\tO(m+n^{1.5})$.
All tools used in this paper are considered combinatorial with \cite{linearflow21} being the only exception, as it uses interior-point methods from continuous optimization.
Unlike other non-combinatorial methods such as fast matrix multiplication, these techniques tend to be fast in practice.
Still, if one is interested in a purely combinatorial algorithm one can replace this bound for weighted graphs with the Goldberg-Rao algorithm~\cite{GR98} that has running time $\tO(min(n^{2/3}, m^{1/2})m\log U)$. The resulting algorithm is slower but still subcubic.

\subsection{Gomory-Hu's algorithm and Partial Trees}
\label{prelim:GH}

First, we give some general definitions that many algorithms that are related to Gomory-Hu trees use.
\paragraph{Partition Trees.}
A \emph{partition tree} $T$ of a graph $G=(V,E)$
is a tree whose nodes $1,\dots,l$ are \emph{super-nodes},
which means that each node $i$ is associated with a subset $V_i\subseteq V$; 
and these super-nodes form a disjoint partition $V=V_1 \sqcup\cdots\sqcup V_l$.
An \emph{auxiliary graph} $G_i$ is constructed from $G$ by merging nodes that lie in the same connected component of $T\setminus \{i\}$. For example, if the current tree is a path on super-nodes $1,\ldots,l$, 
then $G_i$ is obtained from $G$ by merging $V_1\cup\cdots\cup V_{i-1}$ into one \emph{contracted node} and $V_{i+1}\cup\cdots\cup V_l$ into another contracted node. We will use the notations $n_i:=\card{V_i}$, $m_i:=\card{E(G_i)}$, and $n'_i:=\card{V(G_i)}$. Note that $n_i' \geq n_i$ since $V(G_i)$ contains $V_i$ as well as some other contracted nodes, with $n_i'=n_i$ only if the tree is a single node.
The following is a brief description of the classical Gomory-Hu algorithm~\cite{GH61} (see Figure~\ref{Figs:GH_trees_3}).

\paragraph{The Gomory-Hu algorithm.}
This algorithm constructs a 
\ifprocs
cut- equivalent 
\else
cut-equivalent
\fi
tree $\T$ in iterations. 
Initially, $\T$ is a single node associated with $V$ (the node set of $G$), 
and the execution maintains the invariant that $\T$ is a partition tree of $V$.
At each iteration, the algorithm picks arbitrarily two graph nodes $s,t$ 
that lie in the same tree super-node $i$, i.e., $s,t\in V_i$.
The algorithm then constructs from $G$ the auxiliary graph $G_i$
and invokes a \MF algorithm to compute in this $G_i$ a minimum $(s,t)$-cut, denoted $C'$.
The submodularity of cuts ensures that this cut is also 
a minimum $(s,t)$-cut in the original graph $G$, 
and it clearly induces a disjoint partition $V_i=S\sqcup T$ 
with $s\in S$ and $t\in T$. The algorithm then modifies $\T$ by splitting super-node $i$
into two super-nodes, one associated with $S$ and one with $T$,
that are connected by an edge whose weight is the value of the cut $C'$,
and further connecting each neighbor of $i$ in $\T$ 
to either $S$ or $T$ (viewed as super-nodes),
depending on its side in the minimum $(s,t)$-cut $C'$
(more precisely, neighbor $j$ is connected to the side containing $V_j$).

\begin{figure*}[ht]
  \begin{center}
    \includegraphics[width=6.2in]{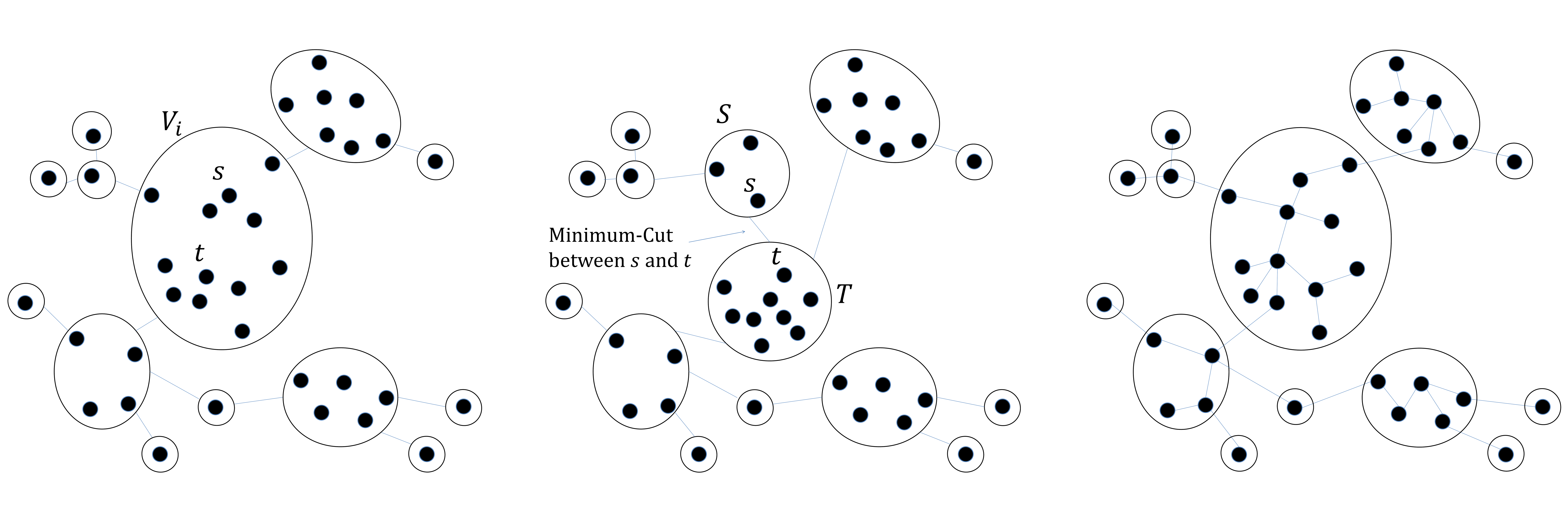}
    \end{center}
    \caption{Illustration of the construction of $\T$.  Left: $\T$ right before the partition of the super-node $V_i$. Middle: after the partitioning of $V_i$. Right: $\T$ as it unfolds after the Gomory-Hu algorithm finishes.}
    \label{Figs:GH_trees_3}
\end{figure*}

The algorithm performs these iterations until all super-nodes are singletons, which happens after $n-1$ iteration.
Then, $\T$ is a weighted tree with effectively the same node set as $G$.
It can be shown \cite{GH61} that for every $s,t\in V$,
the minimum $(s,t)$-cut in $\T$, viewed as a bipartition of $V$,
is also a minimum $(s,t)$-cut in $G$, and of the same cut value.
We stress that this property holds regardless of the choice made at each step
of two nodes $s\neq t\in V_i$. 
A \GHEPT is a partition tree that can be obtained by a truncated execution of the Gomory-Hu algorithm, in the sense that there is a sequence of choices for the pairs $s\neq t\in V_i$ that can lead to such a tree. The following simple lemma describes the flexibility in designing cut-equivalent tree algorithms based on the Gomory-Hu framework.

\begin{lemma}
\label{lem:partial_tree_comb}
Given a \GHEPT $T'$ of an input graph $G$, and a cut-equivalent tree $T_i$ of each auxiliary graph $G_i$ for the super-nodes $V_i$ of $T'$, it is possible to construct a full cut-equivalent tree $T$ of $G$ in linear time.
\end{lemma}
\begin{proof}
In a preprocessing step, for every super node $V_i$ in $T'$ and every contracted node $q\in  V(G_i)\setminus V_i$, save a pointer to the super node $V_j\subseteq q$ that is adjacent to $V_i$ in $T'$.
Now, for every super-node $V_i$, identify each contracted node $q\in V(G_i)\setminus V_i$ with the super-node it contains that is adjacent to $V_i$ in $T'$, and connect the nodes of $V_i$ to the super-nodes they are adjacent to (if any).
Finally, if a node $u\in V_i$ is connected to a super-node $V_j$, and a node $v\in V_j$ is connected to $V_i$, remove these connections and connect $u$ to $v$ directly, and call the result $T$. Observe that $T$ must be a tree.

To see why $T$ is a correct cut-equivalent tree of $G$, observe that there exists a simulated Gomory-Hu execution that results in $T$. Given the \GHEPT $T'$, pick pairs of nodes from $V_i$ and cuts according to $T_i$. This is guaranteed to produce a tree $\tilde{T}_i$ whose projection on $V_i$ is identical to $T_i$, while the subtrees adjacent to $V_i$ in $\tilde{T}$ are connected to the same nodes of $V_i$ as their contracted counterparts in $T_i$. 
Applying this simulated execution to all super-nodes concludes the proof.
\end{proof}

Next, we discuss a few kinds of (GH-Equivalent) partition trees.

\subsubsection{Partial Trees for Subsets}
A \emph{partial tree} for a subset $Q\subseteq V(G)$ is a partition tree $T$ of $G$ such that each super-node $V_i$ in $T$ contains exactly one node from $Q$, and for every two nodes $a,b\in Q$, the minimum cut in $T$ between the super-nodes containing them $A,B$ is a minimum $(a,b)$-cut in $G$. It is folklore that for every subset $Q\subseteq V$, a partial tree for $Q$ can be computed by $\card{Q}-1$ applications of \MF.
\begin{lemma}[see e.g.~\cite{GH86}]\label{lem:partialGH}
Given a graph $G=(V,E,c)$ and a subset $Q\subseteq V$, it is possible to compute a partial tree for $Q$ in time $O(\card{Q}T(m))$, where $T(m)$ is the time for solving (single pair) \MF.
\end{lemma}
This lemma follows by running a Gomory-Hu execution, but always picking pairs from $Q$ (thus making $T$ a \GHEPT).

\subsubsection{$k$-Partial Trees}
A $k$-partial tree, formally defined below, can also be thought of as
the result of contracting all edges of weight greater than $k$
in a cut-equivalent tree of $G$. 
Such a tree can obviously be constructed using the Gomory-Hu algorithm,
but as stated below (in Lemma~\ref{Lemma:Partial}), 
faster algorithms were designed in~\cite{HKP07,BHKP07},
see also~\cite[Theorem $3$]{Panigrahi16}.

It is known (see~\cite[Lemma $2.3$]{AKT20}) that such a tree is a \GHEPT.

\begin{definition}[$k$-Partial Tree~\cite{HKP07}] 
  A \emph{$k$-partial tree} of a graph $G=(V,E)$ is a weighted tree on $l\leq \card{V}$ super-nodes constituting a partition $V=V_1 \sqcup\cdots\sqcup V_l$, with the following property: 
  For every two nodes $s,t\in V$ whose minimum-cut value in $G$ is at most $k$, $s$ and $t$ lie in different super-nodes $s\in S$ and $t\in T$, such that the minimum $(S,T)$-cut in the tree defines a bipartition of $V$
  which is a minimum $(s,t)$-cut in $G$ and has the same value.
\end{definition}

\begin{lemma}[\cite{BHKP07}]
\label{Lemma:Partial}
There is an algorithm that given an undirected graph with $n$ nodes and $m$ edges with unit edge-capacities and an integer $k\in [n]$, constructs a $k$-partial tree in time 
\ifprocs
$\min\{\tO(nk^2),$ $\tO(mk)\}$.
\else
$\min\{\tO(nk^2),\tO(mk)\}$.
\fi
\end{lemma}

\subsubsection{A basic property of the Gomory-Hu tree}

\begin{lemma}\label{Lemma:GH_edges}
Given a graph $G=(V,E)$ and a tree $T$ on the same set of nodes, if for every edge $uv\in T$ the cut $(S_u,S_v)$ resulting from removing $uv$ in $T$ is a minimum cut in $G$, then $T$ is a cut-equivalent tree.
\end{lemma}
\begin{proof}
The proof follows by simulating a Gomory-Hu tree execution with node pairs and minimum cuts taken according to the edges in $T$.
\end{proof}

\subsection{Nagamochi-Ibaraki Sparsification}

We use the sparsification method by Nagamochi and Ibaraki~\cite{NI92}, who showed that for any graph $G$ it is possible to find a subgraph $H$ with at most $k(n-1)$ edges, such that $H$ contains all edges crossing cuts of value $k$ or less.
It follows that if a cut has value at most $k-1$ in $G$ then it has the same value in $H$, and if a cut has value at least $k$ in $G$ then it also has value at least $k$ in $H$.
The authors~\cite{NI92} gave an algorithm that performs this sparsification in $O(m)$ time on unweighted
graphs, independent of $k$. 
They also gave a sparsification algorithm for weighted graphs, with an $O(m + n \log n)$ running time ~\cite{NI92_b}.
In weighted graphs, the sparsification is defined by equating an edge of weight $w$ with a set of $w$ unweighted (parallel) edges with the same endpoints.

\subsection{Latest Cuts}
\label{sec:latest}
Introduced by Gabow~\cite{Gabow91}, a \emph{latest} minimum $(s,t)$-cut (or a \emph{minimal} minimum $(s,t)$-cut, in some literature) is a minimum $(s,t)$-cut $(S_{s},S_{t}=V\setminus S_{s})$ such that no strict subset of $S_{t}$ is a minimum $(s,t)$-cut. It is known that latest minimum cuts are unique, and can be found in the same running time of any algorithm that outputs the maximum network flow between the pair, by finding all nodes that can reach $t$ in the residual graph.
In particular, all upper bound stated in Section~\ref{sec:comb} above for \MF also hold for finding the latest minimum $(s,t)$-cut.

We will use the following properties of cuts.
\begin{fact}[Submodularity of cuts]\label{Fact:Submodularity}
For every two subsets of nodes $A,B\subseteq V$, it holds that $\delta(A)+\delta(B)\geq \delta(A\cup B)+\delta(A\cap B)$.
\end{fact}

\begin{fact}[Posimodularity of cuts]\label{Fact:Posimodularity}
For every two subsets of nodes $A,B\subseteq V$, it holds that $\delta(A)+\delta(B)\geq \delta(A\setminus B)+\delta(B\setminus A)$.
\end{fact}

\begin{lemma}
\label{lem:latest_tree}
Let $G$ be any graph and $p$ be any node. There is a cut-equivalent tree that contains all the latest minimum cuts with respect to $p$.
\end{lemma}
\begin{proof}
This follows because all latest minimum cuts with respect to $p$ form a laminar family, by Fact~\ref{Fact:Submodularity}.
\end{proof}

\begin{lemma}
\label{lem:latest_union}
If $A,B$ are minimum $(p,a)$-cut and minimum $(p,b)$-cut, respectively, and $b \in A$ then $A \cup B$ is a minimum cut for $a$.
\end{lemma}
\begin{proof}
Considering Fact~\ref{Fact:Submodularity}, it only remains to show that $\delta(B)-\delta(A\cap B)\leq 0$. But this is immediate, as $b$ is in both sets, and $B$ is a minimum cut for $b$.
\end{proof}

\begin{lemma}
\label{lem:latest_minus}
If $A,B$ are minimum $(p,a)$-cut and minimum $(p,b)$-cut, respectively, and $a \notin B, b \notin A$ then $A \setminus B$ is a minimum cut for $a$.
\end{lemma}
\begin{proof}
Considering Fact~\ref{Fact:Posimodularity}, it only remains to show that $\delta(B)-\delta(B\setminus A)\leq 0$. But this is immediate, as $b$ is in both sets, and $B$ is a minimum cut for $b$.
\end{proof}

\subsection{Expander Decomposition}
We mostly follow notations and definition from~\cite{SW19}.
Let 
\ifprocs
$\vol_G(C)$ $=\sum_{v\in C}\cdeg_G(v)$
\else
$\vol_G(C)=\sum_{v\in C}\cdeg_G(v)$
\fi
be the \emph{volume} of $C\subseteq V$, where subscripts indicate what graph we are using, and are omitted if it is clear from the context.
The \emph{conductance} of a cut $S$ in $G$ is $\Phi_G(S)= \frac{\delta(S)}{\min(\vol_G(S),\vol_G(V\setminus S))}$.
The expansion of a graph $G$ is $\Phi_{G} = \min_{S\subset V}\Phi_G(S)$. 
If $G$ is a singleton, we define $\Phi_{G}=1$. 
Let $G[S]$ be the subgraph induced by $S\subset V$, and we denote $G\{S\}$ as the induced subgraph $G[S]$ but with added self-loops $e=(v,v)$ for each edge $e'=(v,u)$ where $v\in S,u\notin S$ (where each self-loop contributes $1$ to the degree of a node), so that any node in $S$ has the
same degree as its degree in $G$.
Observe that for any $S\subset V$, $\Phi_{G[S]}\ge\Phi_{G\{S\}}$, because the self-loops increase the volumes but not the values of cuts. We say a graph $G$ is a $\phi$ \emph{expander} if $\Phi_{G}\ge\phi$, and we call a partition $H_1,\ldots,H_h$ of $V$ 
a $\phi$ expander decomposition if $\min_{i}\Phi_{G[V_{i}]}\ge\phi$.

\begin{theorem}[Theorem $1.2$ in~\cite{SW19}]
\label{thm:exp-dec}
Given a graph $G=(V,E)$ of $m$ edges and a parameter $\phi$, there is a randomized algorithm that with high probability finds a partitioning of $V$ into $V_{1},\dots,V_{k}$ such that $\forall i:\Phi_{G[V_{i}]}\ge\phi$
and $\sum_{i}\delta(V_{i})=O(\phi m\log^{3}m)$. In fact, the algorithm
has a stronger guarantee that $\forall i:\Phi_{G\{V_{i}\}}\ge\phi$. The running
time of the algorithm is $O(m\log^{4}m/\phi)$.
\end{theorem}

\section{Motivating Examples and the Isolating-Cuts Procedure}
\label{sec:proc}

This section attempts to explain the thought process that has led to the introduction of two new tools into the context of \GHT algorithms: the \pC  and the \expanders procedures, as well as to give more details about the former while deferring the latter to 
\ifprocs
the full version.
\else
Section~\ref{sec:expanders}.
\fi
Let $G$ be a graph and $T$ be its cut-equivalent tree.
As described in Section~\ref{overview}, previous techniques (Gomory-Hu recursion, partial trees, and randomized pivot selection) can reach the point that $T$, when rooted at a designated pivot $p$, has small sublinear depth.
This suggests that a path-like $T$ is not the hardest to construct since its size must be sublinear.
It is natural to wonder if the other extreme of star-like $T$ is also easy, or perhaps it can lead to a hardness reduction.

The most extreme case is when $T$ is precisely a star with $p$ at the center and all other nodes as leaves. 
Constructing this tree is easy because the minimum $(p,v)$-cut for all $v \in V$ is the singleton $(\{v\},V\setminus \{v\})$ and its value is $\deg(v)$.
(However, verifying that the star is the correct tree is perhaps as hard as the general problem.)

The following two examples show that by changing the star slightly one reaches challenging scenarios that require new tools.

\subsection{Many small subtrees}\label{sec3:motivation}
Consider a graph whose cut-equivalent tree $T$ is a ``star of triples graphs'' as depicted in Figure~\ref{Figs:Triangles}. In this example, $T$ is simply a star on $n$ nodes where each leaf is connected to additional two leaves, and so altogether the number of nodes is $3n+1$. Call the center node $p$, the inner nodes $L_1=\{u_1,\dots,u_n\}$, and the leaves $L_2$.
Assuming that $p$ and $L_1$ are given in advance, how fast can we identify which pair of leaves belongs to which inner node in the cut-equivalent tree (thus constructing $T$)?
Note that it can be done by simply asking a \MF query between $p$ and  each node $u \in L_1$, but can it be done faster than $\Omega(n)$ applications of \MF?\footnote{Readers familiar with \APSP-hardness results may recall that triangle identification is at the core of most reductions, making this example appealing.}

\begin{figure}[ht]
  \begin{center}
    \includegraphics[width=4.0in]{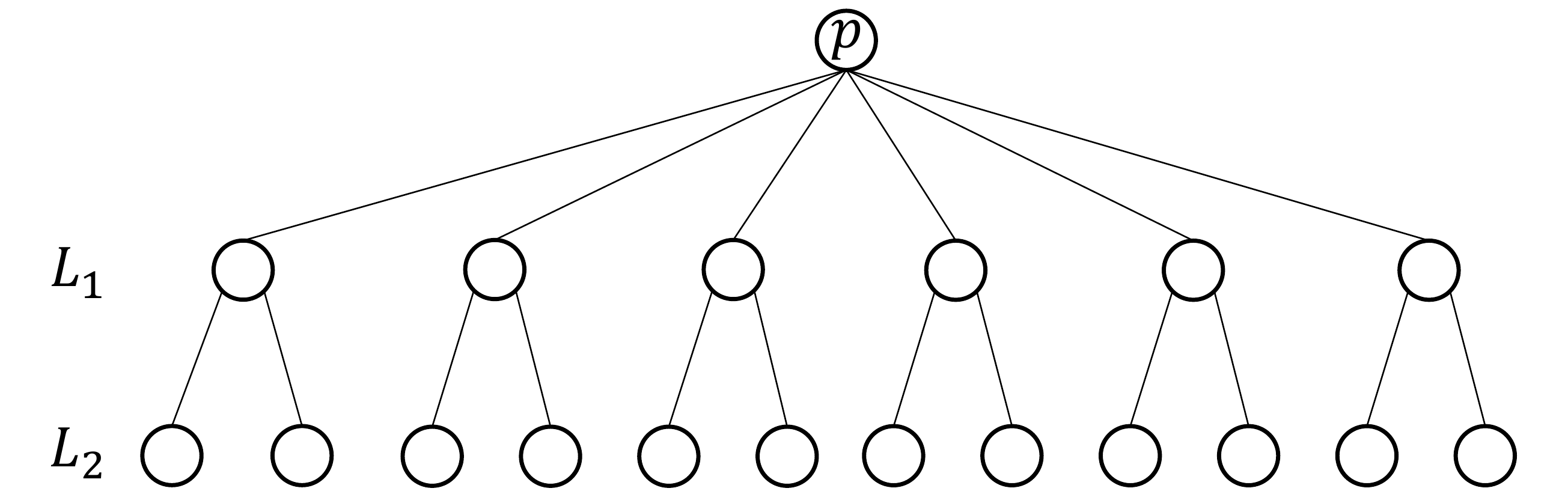}
    \end{center}
    \caption{The cut-equivalent tree $T$.}
    \label{Figs:Triangles}
\end{figure}

It turns out that the answer is yes!  
A simple but surprising algorithm computes $T$ using only $\tO(1)$ queries to a $\MF$ algorithm (on a weighted graph).

\paragraph{The Algorithm}
\begin{enumerate}
\item Add edges of very large capacity $U:=|E|^2$ from $p$ to every node in $L_1$. 
Denote the new graph by $G_m$ and its cut-equivalent tree by $T_m$. 
\item
Repeat $10\log^2 n$ times:
\begin{enumerate}
\item
In $G_m$ split $p$ to $p_1$ and $p_2$, connecting the (previously connected to $p$) newly added edges 
to $p_1$ or $p_2$ with probability $1/2$, and the rest of the edges arbitrarily. 
Then add an edge of weight $U^2$ between $p_1$ and $p_2$. 
Denote the new graph $G_h$ and its cut-equivalent by $T_h$.
\item
Find the minimum $(p_1,p_2)$-cut in $G_h$.
\end{enumerate}
\item
Construct the new tree $T'$ as follows. For every node $u\in L_1$ connect it in $T'$ to the leaves in $L_2$ that went to the same side as $u$ in all repetitions.
\item
If any node in $L_1$ has a number of leaves that is different than two, output ``failure'', otherwise report $T'$. 
\end{enumerate}

\paragraph{Why does it work?}
Consider $G_m,T_m,G_h,G_h$ in one iteration. 
\begin{observation}\label{Obs:overview_T_m}
The tree $T_m$ is identical to $T$, except for the weights of edges adjacent to $p$.
\end{observation}
For every edge $\{u,v\}$ in $T$, consider two cases. If $u$ or $v$ is $p$, then the new edge increases all $(u,v)$-cuts by at least $U$, and since the old cut increases by exactly $U$, it is a minimum cut in $T_m$ as well.
Otherwise, if neither of $u$ nor $v$ is $p$, then the minimum $(u,v)$-cut is still the same because the new edges did not affect this cut, and can only increase the weights of other cuts.
Thus, the minimum $(u,v)$-cut in each case is the same in $T$ and $T_m$.
\begin{observation}\label{Obs:overview_T_h}
There is an edge between $p_1$ and $p_2$ in $T_h$. If it is contracted $T_h$ becomes identical to  $T_m$.
\end{observation}
This is true because $p_1$ and $p_2$ are on the same side of every minimum $(a,b)$-cut for any pair $\{a,b\} \neq \{p_1,p_2\}$ because the weight of the edge between them is larger than the sum of all others.

The following is the key claim in this section.
\begin{claim}
The minimum $(p_1,p_2)$-cut in $G_h$ sends each pair of siblings in $L_2$ to the same side as their parent in $L_1$.
\end{claim}
\begin{proof}
Let $u_1,u_2\in L_2$ be a pair of siblings whose father is $u \in L_1$. By Observations~\ref{Obs:overview_T_m} and~\ref{Obs:overview_T_h}, $u_1$ and $u_2$ are $u$'s neighbors in $T_h$. 
Consider $T_h$ after the removal of the edge $\{p_1,p_2\}$: the two resulting subtrees are the two sides of the minimum $(p_1,p_2)$-cut in $G_h$.
Since $u,u_1,u_2$ are connected in $T_h$, it follows that they must be on the same side of the cut.
\end{proof}

Hence, every node $u\in L_1$ can know its children with high probability by seeing which two nodes in $L_2$ are always sent with it to the same side of the minimum $(p_1,p_2)$-cuts throughout all iterations. 

\subsection{The Isolating-Cuts Procedure}

The above algorithm can be generalized to obtain the \pC procedure that is a key ingredient of the new algorithm.
Notably, the added weights are not necessary (can be replaced with contractions) nor is the randomness (can be replaced with deterministic separating choices).
 Let $MF(N,M,F)$ be an upper bound on \MF in graphs with $N$ nodes, $M$ edges, and where the flow size is bounded by $F$.
We will utilize the following statement that essentially follows from the very recent work of Li and Panigrahy \cite[Theorem II.2]{LP20} for \GMC. 
We provide another proof in 
\ifprocs
the full version
\else
Section~\ref{app:proc} 
\fi
both for completeness and because it could be of interest as it exploits the structure of the \GHT (as we did above) instead of using the submodularity of cuts directly. In particular, we use ideas similar to Lemma~\ref{Lemma:GH_edges}.

\begin{lemma}[The \pC Procedure]
\label{thm:proc_main}
Given an undirected graph $G=(V,E,c)$ on $n$ nodes and $m$ total edges, a \emph{pivot} node $p\in V$, and a set of \emph{connected} vertices $C \subseteq V$, let $(C_v,V\setminus C_v)$ where $v \in C_v, p \in V\setminus C_v$ be the latest minimum $(p,v)$-cut for each $v \in C$.
There is a deterministic $O(MF(n,m,c(E)) \cdot \log{n})$-time algorithm that returns $|C|$ disjoint sets $\{ C'_v \}_{v \in C}$ such that
for all $v\in C$: if $C_{v} \cap C = \{v\}$ then $C'_v = C_v$.
\end{lemma}

\subsection{A Single Large Subtree}
\label{sec:clcr}

To solve the previous example, we have exploited the fact that the identity of the nodes in $L_1$ is known to us in advance. 
However, by picking each node to be ``connected'' (and added to the set $C$) with probability $1/2$ there is a $1/8$ chance that a parent $u\in L_1$ is in $C$ while its two children $u_1,u_2 \in L_2$ are not.
Then, applying Lemma~\ref{thm:proc_main} $O(\log n)$ times is sufficient for constructing $T$.
More generally, if all nodes in $T$ have a subtree that is small, e.g. with $n^{\eps}$ descendants, then $n^{\eps}$ repetitions of the \pC procedure with random choices for $C$ ought to suffice. 

This leads to the second hard case where there is a single (unknown) node with a subtree of size $\Omega(n)$.
Consider the example in Figure~\ref{Figs:clcr} where $T$ is composed of a left star centered at $c_\ell$ and a right star centered at $c_r$, and suppose that $c_r$ is the (known) pivot $p$ while the challenge is to identify $c_\ell$.
The minimum $(c_\ell,c_r)$-cut is large and balanced $(C_\ell,C_r)$ while nearly all other minimum $(u,v)$-cuts in the graph are the trivial singleton cuts.
The weights on $T$ can be set up so that $c_\ell$ is almost like a needle in a haystack: asking almost any $(v,c_r)$ \MF query for any $v \neq c_\ell$ may not reveal any information about $c_\ell$. Suppose that $\lambda_{c_\ell,c_r}$ is $w$, about half of the edges between $c_r$ and its children are $>w$ and half are $<w$, while nearly all nodes on the left have an edge of weight $<w$ to $c_\ell$.
Thus, even if a randomly chosen pair $u,v$ will have one node on each side, the minimum $(u,v)$-cut will be either $(\{u\},V\setminus \{u\})$ or $(\{v\},V\setminus \{v\})$ with high probability, which does not help discover $c_\ell$ nor any of its cut members $C_\ell$.

\begin{figure}[ht]
  \begin{center}
    \includegraphics[width=2.5in]{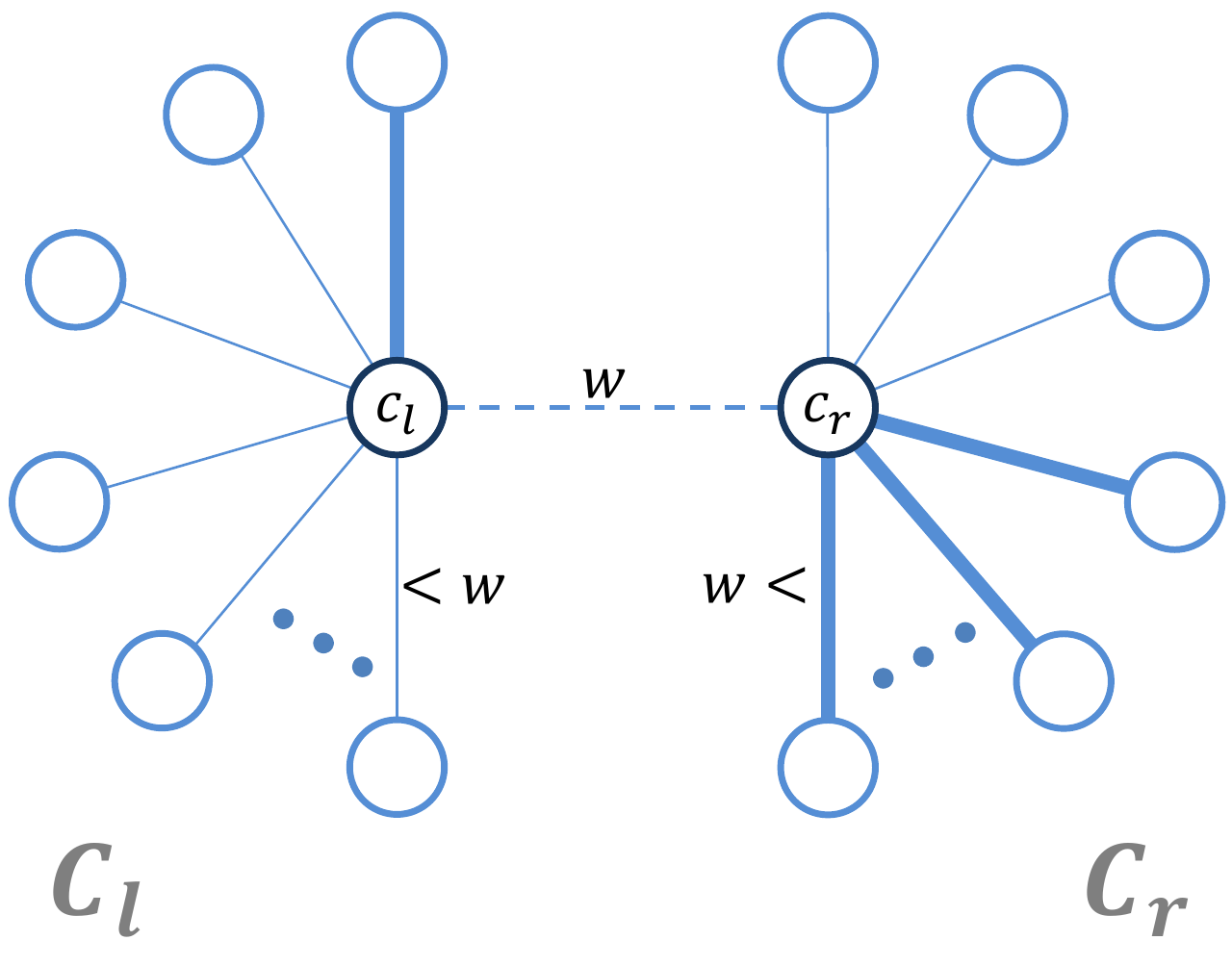}
    \end{center}
    \caption{ The cut-equivalent tree $T$ of a hard case where $c_r$ is given but $c_\ell$, the only node with a non-trivial minimum cut to $c_r$, must be identified.
    The weight of the $(C_\ell,C_r)$ cut is $\lambda_{c_l,c_r} = w$, illustrated by the dashed edge at the center, and the edges of weight $>w$ are thick, while edges of weight $<w$ are thin. The minimum cut between any pair of leaves $u,v$ is trivial unless $u$ is attached to one of the bold edges on the left and $v$ is attached to one of the bold edges on the right.
    }
    \label{Figs:clcr}
\end{figure}

To apply the \pC procedure one must find a way to \emph{isolate} $c_\ell$, i.e. to choose a set of connected nodes $C$ such that $c_\ell$ is in $C$ but none of its $\Omega(n)$ cut-members are in $C$.
This is an impossible task \emph{unless} the structure of the graph $G$ is exploited.
The \expanders procedure manages to do just that, as overviewed in Section~\ref{overview} and fully described in 
\ifprocs
the full version,
\else
Section~\ref{sec:expanders}, 
\fi
by exploiting a decomposition of $G$ to either reach a good $C$ or to directly query the pair $c_\ell,c_r$.

Notably, this example remains the bottleneck after the new algorithm and might lead the way to the first conditional lower bound for \GHT and \APMF. 

\ifprocs\else 

\section{The Main Algorithm}
\label{Section:Cut_Alg}

This section proves the main results of the paper by describing and analyzing the new algorithm and how it can be adjusted to give different bounds.
The two most interesting ingredients of the algorithm are presented and analyzed in other sections.
All other results stated in the introduction directly follow from Theorem~\ref{mainthm} by setting $M=N^2$.
The rest of this section is dedicated to its proof.

\medskip

Given a simple graph $G=(V,E)$ on $N$ nodes and $M$ edges, the algorithm below constructs a cut-equivalent tree $T_G$ for $G$.

\paragraph{Parameter Selection and Preliminaries.}
Let $\gamma$ be a constant that controls the success probability of the algorithm.
There are three other important parameters that are set based on the \MF algorithms that the algorithm is allowed to use.
\begin{enumerate}
\item To get bound~\ref{Bound:Existing:Algebraic:Theorem}, the $\tO(\min(M^{3/2}N^{1/6} , \max(MN^{3/4},M^{3/2})))$-time algorithm from~\cite{AKT20} is sufficient when $M< N^{3/2}$. Therefore, the new algorithm is only used when $M \geq N^{3/2}$ and the following selections are made: Set $r:=\sqrt{M/N}$, $k:=\sqrt{N}$, and for $k\leq w\leq N$ define $\phi_w := \sqrt{N/M}$.

\item To get bound~\ref{Bound:Existing:Combinatorial:Theorem} the combinatorial $O(MN^{2/3})$ Max-Flow algorithm~\cite{GR98} can be used as a subroutine in the algorithm from~\cite{AKT20} to get a $\tO(M^{3/2}N^{1/3})$-time combinatorial algorithm for the case $M< N^{4/3}$. Therefore, the new algorithm is only used when $M \geq N^{4/3}$ and the following selections are made. 
Set $r:= \sqrt{M}/N^{1/6}$, $k:=\sqrt{M}\log^3 N/N^{1/6}$, and for $k\leq w\leq N$ define $\phi_w := \sqrt{M}/(N^{1/6}\cdot w )$.
\end{enumerate}

Note that our choices of parameters satisfy two assumptions that are required in order to use the 
\expanders procedure of Section~\ref{sec:expanders} (for technical reasons):
\begin{itemize}
\item It holds that $\phi_w < 1/\log^3 N$ whenever $w \geq k$. This is important because the expander decomposition algorithm requires the parameter $\phi$ to be $O(1/\log{N})$, and furthermore it must be $O(1/\log^3 N)$ for the outside edges bound to be meaningful. 
\item It holds that $r \geq 1/\phi_w$ whenever the new algorithm is used (i.e. except in the sparse regimes). 
\end{itemize}
In Section~\ref{sec:expanders} we will prove both the correctness and running time bounds of the \expanders procedure for this specific choice of the parameters.\footnote{We apologize for the inconvenience of the forward referencing between Sections~\ref{Section:Cut_Alg} and~\ref{sec:expanders}. Changing the order of sections would not resolve the issue because each of them makes certain choices based on the other.}

\paragraph{Preprocessing}

\begin{itemize}

\item First, compute a $k$-Partial Tree of $G$ in time $O(Nk^2)$ (see Lemma~\ref{Lemma:Partial}).
The next steps show how to handle each of the super-nodes where for any pair $u,v$ in the super-node we have $\lambda_{u,v} \geq k$.

\item For each $w=2^i$ where $ i \in \{ \lfloor \log k \rfloor ,\ldots, \lceil \log{N} \rceil \}$ we do the following:

\begin{itemize}
\item Compute a Nagamochi-Ibaraki sparsifier $G_w$ of $G$ with $O(Nw)$ edges such that all cuts of weight up to $2w-1$ are preserved and all cuts of weight at least $2w$
still have weight $\geq 2w$. The running time is $O(M)$.

\item Compute an expander-decomposition of $G_w$ into expanders $H^{(w)}_1,\ldots,H^{(w)}_h$ with parameter $\phi_w$, as in Theorem~\ref{thm:exp-dec}. 
%
By the parameter selection above, it holds that $\phi_w< 1/\log N$ whenever $w\geq k$ and so the theorem can indeed be used. Each expander has expansion $\geq \phi_w$ and the total number of edges outside the expanders $\sum_{i=1}^h \delta(H^w_i)$ is $\tilde{O}(Nw \phi_w)$. The running time is $\tilde{O}(Nw/\phi_w)$. 
\end{itemize} 

\end{itemize}

\paragraph{The Recursive Algorithm}
Suppose $T$ is \GHEPT of $G$ that we wish to complete into a full cut-equivalent tree, and let $V_i$ be one of $T$'s super-nodes with its corresponding auxiliary graph $G_i$, where $|V_i|=n_i', |V(G_i)|=n_i$ and $|E(G_i)|=m_i$.
(Initially, $V_i=V$.) 
The following steps show how to ``handle'' $V_i$: we will return a tree $T_i$ that is a cut-equivalent tree for $G_i$ using calls to the \pC and \expanders procedures, and by making recursive calls on super-nodes $V_{i,1},\ldots,V_{i,k}$ that are disjoint and that have size $|V_{i,j}| \leq |V_i|/2$.

\medskip

\begin{enumerate}

\item\label{Recursive_Algorithm:Resolving} If $|V_i| < r$, resolve $V_i$ by executing the Gomory-Hu algorithm. Thus the time is $n_i-1$ applications of \MF, which is at most $O(r\cdot MF(n_i,m_i,N))$ since $N$ is an upper bound on any Max-Flow in $G$. 

\item\label{Recursive_Algorithm:Sampling} Pick a random subset $S \subseteq V_i$ by picking a uniformly random node from $V_i$ repeatedly $n'_i/r  \cdot (\gamma + 2)\log{N}+1 = \tilde{O}(n_i/r)$ times. 
Use Lemma~\ref{lem:partialGH} to compute a partial tree $T_S$ for $S$ using the Gomory-Hu algorithm.
This partitions $V_i$ into $t$ super-nodes $V_{i,1},\ldots,V_{i,t}$, connected in a partial tree $T_S$, such that each part contains exactly one node from $S$. We recurse on each super-node $V_{i,j}$ that has size $\leq n_i/2$, by defining its auxiliary graph $G_{i,j}$ using the partial tree, and obtain its cut-equivalent tree $T_{i,j}$.
Note that there could only be one super-node $V'=V_{i,j}$ whose size is $>n_i/2$ and so we have to handle it more carefully using the steps below.
Let $G'$ be its auxiliary graph (that can be obtained from the partial tree); our goal is to compute a cut-equivalent tree $T'$ of $G'$.
Once we have $T'$ and all the $T_{i,j}$'s, we simply ``fill them into'' the partial tree $T_S$ (see Lemma~\ref{lem:partial_tree_comb}) to get the final $T_i$.

From now on we only focus on the super-node $V'$, with its auxiliary graph $G'$, where $n':=|V'| \leq n'_i \leq N, n:=|V(G')| \leq n_i \leq N, m:= |E(G')| \leq m_i \leq M$. 
Let $p \in S$ be the node from $S$ that is in $V'$; this will be our \emph{pivot}.
The way $p$ was selected guarantees that $T'$ has two important properties with high probability (see Claims~\ref{cl:depth} and~\ref{cl:problematic} below): the ``depth'' is $O(r)$ and there are only $O(r)$ ``problematic'' vertices $v$ whose latest minimum $(p,v)$-cut is not small enough (in terms of $|C_v \cap V'|$).
The running time of this step is $\tilde{O}(n_i/r \cdot MF(n_i,m_i,m_i))$.

\item\label{Recursive_Algorithm:Maintaining} For each node $v \in V'$ we maintain and update an estimate $c'(v)$ for its minimum $(p,v)$-cut \emph{value} in $G'$, as well as a testifying cut $(S_v,V(G')\setminus S_v)$ whose value is equal to $c'(v)$. We initialize $c'(v)$ to be $\deg_{G}(v)$ for all nodes $v$ since this is an upper bound on the true minimum cut value to any other vertex and in any auxiliary graph obtained from a \GHEPT of $G$. This initial value is testified by the cut $(\{v\},V(G')\setminus \{v\})$. As we progress and find new cuts that separate $v$ from $p$, these estimates could decrease until they become equal to the correct value $\lambda_{p,v}$. Due to the way we update them, these estimates never decrease below $\lambda_{p,v}$.

\item\label{Recursive_Algorithm:Expander_Querying} Repeat the \expanders procedure defined in Section~\ref{sec:expanders}, $\log n+2$ times. The goal of each repetition is to compute new estimates $c'(v)$ for all nodes. After these $O(\log{n})$ repetitions, all estimates will be correct with high probability (proved in Claim~\ref{cl:repetitions}) and the testifying cuts are the \emph{latest} minimum $(p,v)$-cut for all $v \in V'$. 

\item\label{Recursive_Algorithm:Finishing_up} At this point, with high probability: for all nodes $v\in V'$ the estimates are correct $c'(v)=\lambda_{v,p}$ and we have the latest minimum $(p,v)$-cut $(C^p_v,V(G')\setminus C^p_v)$ with $v\in C^p_v$. (Note that, while this is a lot of information, we are not done with computing $T'$ because these cuts may not determine the optimal cut for some pairs $u,v \neq p$.) 
Since they are latest, we know that these cuts are non-crossing, and by Claim~\ref{cl:problematic} we know that only $r/2$ of them can have $|C_v^p \cap V'|>n'/2$.

Order all nodes $v \in V'$ by the size of $C^p_v \cap V'$, from largest to smallest, and go over this list:
First, place all nodes $v \in V'$ such that $|C^p_v\cap V'| > n'/2$ in a set $P$ and remove them from the list.
Note that by Claim~\ref{cl:problematic}, $|P|\leq r/2$.
Next, pick the remaining node $v$ with largest $|C^p_v\cap V'|$, define $V'_1 := C^p_v$ to be its side in the cut, and remove all nodes $u \in V' \cap C^p_v$ from the list. 
Repeat this until all nodes are removed from the list to get subsets $V'_1,V'_2,\ldots,V'_t \subseteq V(G')$.
Since these subsets correspond to latest cuts with respect to a single pivot $p$, they must be disjoint (Lemma~\ref{lem:latest_tree}).
Let $P'=V(G') \setminus \cup_{i=1}^t V'_i$ be the set of nodes (and contracted nodes) outside these sets, and note that $P = P' \cap V'$.

Construct $T''$ that is a \GHEPT of $G'$ as follows: 
The super-nodes are $V'_1,V'_2,\ldots,V'_t$ and $P'$.
The tree is a star with $P'$ at the center and $V'_1,V'_2,\ldots,V'_t$ are the leaves.
The weight of the edge between $P'$ and $V'_i$ is set to $c'(v_i)=\lambda_{v_i,p}$ the weight of the cut $(V_i',V(G)\setminus V_i')$.

Then, we complete $T''$ into the final tree $T'$ by making recursive calls for each one of its super-nodes, defining its auxiliary graph based on $T''$.
We combine the trees of the super-nodes as in Lemma~\ref{lem:partial_tree_comb} to get $T'$, and we are done.
The running time of this step is $O(n^2)$ plus the time for the recursive calls.
\end{enumerate}

\subsection{Running Time Analysis}

First, we show that for every partition tree, the total sum of nodes in all auxiliary graphs is bounded, which we will use later.
\begin{lemma}
\label{lem:Nodes_bound}
For every partition tree $T_{part}$ of a graph $G=(V,E)$, the total sum of nodes in all auxiliary graphs of the super-nodes of $T_{part}$ is $O(|V|)$.
\end{lemma}

\begin{proof}
Observe that the set of nodes in all auxiliary graphs is comprised of nodes appearing as themselves (not as contracted nodes) which totals to exactly $\card{V}$, and contracted nodes whose number for every auxiliary graph $G_i$ is exactly $\deg_{T_{part}}(V_i)$, and thus totals to the sum of degrees in $T_{part}\leq O(\card{V})$, concluding Lemma~\ref{lem:Nodes_bound}.

\end{proof}

We continue to bound the running time of the recursive algorithm.
To analyze the recursion we will use the fact that the total number of edges across all auxiliary graphs at a single depth is $O(M)$ (see Lemma $3.7$ in ~\cite{AKT20}), and the total number of nodes there is $\tO(N)$ (by Lemma~\ref{lem:Nodes_bound}).
We will only describe the items from the recursive algorithm that are significant in size compared to our bounds.

\begin{enumerate}[i]
\item
In the preprocessing step we spend $\tO(Nk^2)$ time for the $k$-partial tree, and pick $k=\sqrt{N}$ for bound~\ref{Bound:Existing:Algebraic:Theorem} thus spending $\tO(N^{2})$ time in this case, and $k=\sqrt{M}\log^3 N/N^{1/6}$ for bound~\ref{Bound:Existing:Combinatorial:Theorem}, thus spending time $\tO(MN^{2/3})$ in this case.

\item
Step~\ref{Recursive_Algorithm:Resolving} takes $r \cdot \MF(n_i,m_i,N)$, which we analyze in parts, according to the bound we seek.
\begin{enumerate}
\item 
For bound~\ref{Bound:Existing:Algebraic:Theorem}, $r=\sqrt{M/N}$, and so this step incurs a total time of \newline
$\tO\left(\sum_i \sqrt{M/N} \cdot n_iN\right)\leq
\tO(N^{1+1/2}\cdot \sqrt{M})$, using the Karger-Levine algorithm~\cite{KL15}.

\item
For bound~\ref{Bound:Existing:Combinatorial:Theorem}, $r=\sqrt{M}/N^{1/6}$, and so this step incurs a total time of \newline
$\tO\left(\sum_i \sqrt{M}/N^{1/6}\cdot n_i N\right)\leq
\tO(N^{1+5/6}\cdot \sqrt{M})$, similarly using the Karger-Levine algorithm~\cite{KL15} (that is combinatorial).

\end{enumerate}

\item
Step~\ref{Recursive_Algorithm:Sampling} takes $n_i/r \cdot \MF(n_i,m_i,m_i)\cdot (\gamma+2)\log N$, which we analyze in parts, according to the bound we seek.
\begin{enumerate}
\item 
For bound~\ref{Bound:Existing:Algebraic:Theorem}, $r=\sqrt{M/N}$, and so this step incurs a total time of \newline
$\tO\left(\sum_i n_i\cdot \sqrt{N/M}\cdot (m_i+n_i^{3/2})\right)\leq
\tO(N^{3/2}\cdot \sqrt{M}+N^3/\sqrt{M})\leq N^{3/2}\cdot \sqrt{M}$ (since our regime is $M\geq N^{3/2}$), using the recent \MF algorithm~\cite{linearflow21}.

\item
For bound~\ref{Bound:Existing:Combinatorial:Theorem}, $r=\sqrt{M}/N^{1/6}$, and so this step incurs a total time of \newline
$\tO\left(\sum_i n_i\cdot N^{1/6}/\sqrt{M}\cdot m_i n_i^{2/3}\right)\leq
\tO(N^{1+5/6}\cdot \sqrt{M})$, using the combinatorial Goldberg-Rao algorithm~\cite{GR98}.

\end{enumerate}

\item
For Step~\ref{Recursive_Algorithm:Expander_Querying}, we analyze the contribution of the \expanders Procedure in parts, using our bounds from Lemma~\ref{Lemma:Bounds}.
\begin{enumerate}
\item
For bound~\ref{Bound:Existing:Algebraic:Theorem}, we use the bound from~\ref{Bound:Existing:Algebraic:Densest:Lemma} to claim that the total running time over all auxiliary graphs is:
$$
\sum_i \tO( \sqrt{N}\cdot \sqrt{M} \cdot n ) + 
\sum_i \tO( N^{3/2}/\sqrt{M} \cdot m ) +
\sum_i \tO( N^{3/2}/\sqrt{M} \cdot n^{3/2} )
$$ 
$$
\leq
\tO( \sqrt{N}\cdot \sqrt{M} \cdot N + 
N^{3/2}/\sqrt{M} \cdot M +
N^3/\sqrt{M} )
\leq 
\tO( N^{3/2}\sqrt{M} ),
$$
where the last inequality is again due to the fact we only deal with $M\geq N^{3/2}$, as required.
\item
Similarly, for bound~\ref{Bound:Existing:Combinatorial:Theorem}, we use the bound from~\ref{Bound:Existing:Combinatorial:Densest:Lemma} to bound the total running time over all auxiliary graphs:
$$
\tO\left( \sum_i N^{7/6}/\sqrt{M} \cdot m_i n_i^{2/3} + 
\sum_i N^{5/6}\cdot\sqrt{M}\cdot n_i  \right)
$$ 
$$
\leq
\tO(N^{7/6}/\sqrt{M}\cdot MN^{2/3}+ 
N^{5/6}\cdot\sqrt{M}\cdot N) \leq 
\tO(N^{1+5/6}\sqrt{M}),
$$
as required.
\end{enumerate}

\end{enumerate}

\subsection{Correctness Analysis}

Suppose that $T$ is a \GHEPT for $G$ and that super-node $i$ contains nodes $V_i \subseteq V(G)$ and its auxiliary graph is $G_i$.
The goal of this subsection is to prove that the recursive algorithm computes a correct cut-equivalent tree for $G_i$.

The proof is by induction on the number of nodes $|V_i|$ in the super-node. When $|V_i|=1$ the singleton is indeed the correct tree. Suppose that the algorithm works correctly for any $V_j$ with $|V_j|< |V_i| $, and we will prove that the output for $G_i$ is correct.

The cut-equivalent tree that we return will depend on the randomness in choosing the subset $S$; in particular, it determines the pivot $p$.
The plan for the proof is as follows. 
First, we analyze the properties of the pivot $p$. 
Then, we use this to analyze the outcome of the \expanders procedure.
Finally, we show that the latter is sufficient to produce a cut-equivalent tree for $G_i$.

\paragraph{Shallow cut-membership tree with respect to $p$.}
After picking a subset $S$ at random and computing its partial tree $T_S$, the super-node $V'$ and its auxiliary graph $G'$ are determined, where $V' \cap S = \{ p \}$. It is a super-node in a \GHEPT that is a refinement of $T$.

The algorithm returns a very specific cut-equivalent tree $T'$ of $G'$: the tree such that for all $v \in V'$ the (latest) minimum $(p,v)$-cut in $T'$ is $(C_v^p,V(G')\setminus C_v^p)$ the latest minimum $(p,v)$-cut in $G'$.
By Lemma~\ref{lem:latest_tree} such a cut-equivalent tree indeed exists.
The analysis will rely on the fact that $T'$ must have small ``depth''; intuitively, this is because all long paths would have been hit by the random sample $S$ and (partly) removed from $V'$.
The notion of depth that we use is not the standard one and it requires the following notion of \emph{cut-membership tree}; a useful tool for analysis introduced in \cite[Section $3$]{AKT20_b}.  

The \emph{cut-membership tree} of $T'$ with respect to the pivot $p$, denoted $\TG_p$, is a coarsening of $T'$ such that nodes $v\in B$ with the same latest minimum $(p,v)$-cut in $T'$ (and therefore in $G'$) are merged into one bag $B$.
More formally, define the function $\ell:V(T')\setminus\set{p}\rightarrow E({T}')$ such that $\ell(u)$ is the latest lightest edge in the path between $u$ and $p$ in ${T}'$, and $\ell(p)=\emptyset$; where latest means closest to $u$.
Let $\TG_{p}$ be the graph constructed from ${T}'$
by merging nodes whose image under $\ell$ is the same.
Observe that nodes that are merged together,
namely, $\ell^{-1}(e)$ for $e\in E({T}')$, are connected in ${T}'$,
and therefore the resulting $\TG_{p}$ is a tree.
See Figure~\ref{Figs:Tzeta} for illustration. 
Note that $V(\TG_p)$ the nodes of $\TG_{p}$ are subsets of $V(G')$ and we refer to them as \emph{bags}.
For example, $p$ is not merged with any other node, and thus forms its own bag. It is helpful to treat $\{ p \}$ as the root of $\TG_{p}$.
We define the \emph{weight} of a bag $B\in V(\TG_{p})$ to be the number of nodes (but not contracted nodes) in the bag, i.e. $w(B) = | \{ v\in V' \cap B\}|$. 
For a set of bags $X \subseteq V(\TG_{p})$ denote by $w(X)=\sum_{B \in X} w(B)$ the total weight of bags in $X$.
The \emph{weighted-depth} of a bag $B \in V(\TG_{p})$ is the weight of the set of bags on the path between $\{p\}$ and $B$ in $\TG_{p}$, including $B$ but excluding $\{p\}$.
Finally, the weighted-depth of the tree $\TG_p$ is the maximum weighted-depth of any bag $B \in V(\TG_{p})$.

%
\begin{figure*}[!ht]
       \includegraphics[width=0.7\textwidth,center]{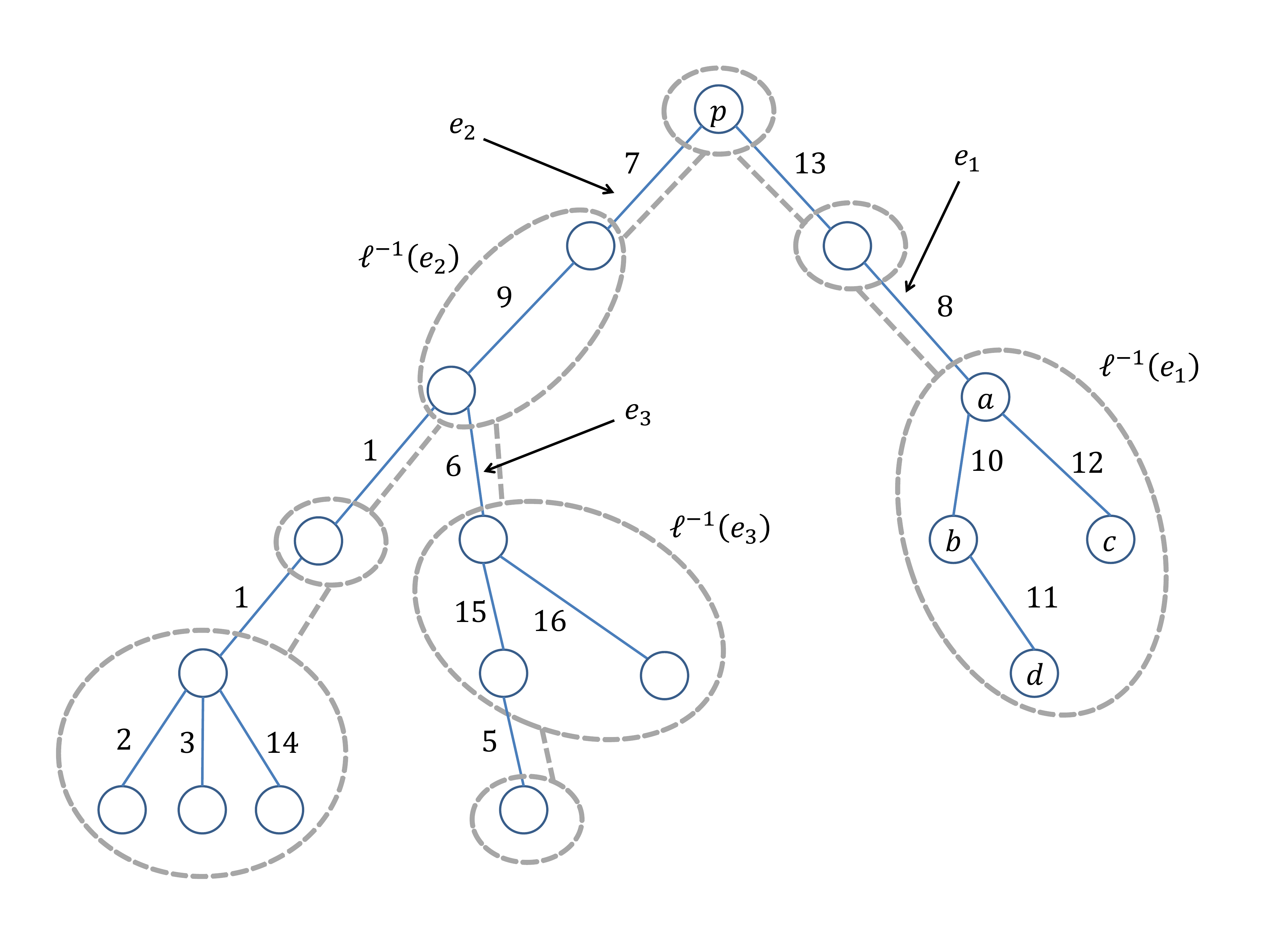}
   \caption[-]{
   An illustration showing ${T}'$ with solid blue lines, while the corresponding graph $\TG_{p}$ with dashed gray lines. For example, $e_1=\ell(a)=\ell(b)=\ell(c)=\ell(d)$ and therefore $\{a,b,c,d\} \in V(\TG_p)$ is a bag of weight $4$. 
   }
   \label{Figs:Tzeta}
\end{figure*}

\begin{claim}
\label{cl:depth}
The weighted-depth of $\TG_p$ is at most $r/2$ with probability at least $1-1/N^{\gamma}$.
\end{claim}

\begin{proof}
Consider any node $q\in V_i\cap S$ that was chosen to $S$.
For each $u \in V_i$, let $(C^q_u,V(G_i)\setminus C^q_u)$ be the minimum $(q,u)$-cut in ${G_i}$ that is latest with respect to $q$, i.e. where $q\in C^q_u$ and minimizing the size of $u$'s side.
For each $x \in V_i$, let $M_x^q = \{ u \in V_i, u \neq q \mid x \in C^q_u \}$ be the set of nodes $u$ such that $x$ is a cut-member of $u$ (with respect to $q$).
The probability that $x$ is in the same super-node as $q$ in the partial tree $T_S$ is upper bounded by the probability that $M_x^q \cap S = \emptyset$, i.e. $(1-|M_x^q|/n'_i)^{|S|-1}$.
To see this, assume for contradiction that there is a node $u \in M_x^q \cap S$, and notice that the minimum $(u,q)$-cut in $T_S$ must contain the latest minimum $(u,q)$-cut $C_u^q$ and therefore $x \in C_u^q$ must be in $u$'s super-node or another super-node on the same side in $T_S$, but not $q$'s.
For a pair $q,x$ such that $|M_x^q|>r/2$ we can upper bound this probability by 
$$
(1-r/2n'_i)^{|S|-1} = (1-r/2n'_i)^{n'_i/r\cdot (\gamma+2) \log N}  \leq 1/N^{\gamma+2}.
$$
By a union bound, the probability that there exists a pair $q \in S, x \in V_i$ in the same super-node of $T_S$ such that $|M_x^q|>r/2$ is at most $1/N^{\gamma}$.

Our special pivot node $p$ is simply the node $q \in S$ with the largest super-node in $T_S$ namely $V'$. 
Therefore, with probability at least $1-1/N^{\gamma}$, all nodes $x \in V'$ satisfy $|M_x^p| \leq r/2$, and due to the following observation, it follows that the weighted-depth of $\TG_p$ is at most $r/2$.

\begin{observation}
\label{obs1}
If $x\in V'$ is not in the latest minimum $(p,u)$-cut in $G_i$ for $u \in V'$, i.e. $u \notin M^p_x$, then $x$ is also not in the latest minimum $(p,u)$-cut in $G'$.
\end{observation}

It follows that for any node $x$ in bag $B_x \in V(\TG_p)$, all nodes on the path to $p$ in $\TG_p$ must be in $M^p_x$ and therefore the weighted depth is at most $r/2$.

To prove Observation~\ref{obs1}, assume for contradiction that there is a node $u \in V'$ such that $x \notin C^p_u$ (the latest cut in $G_i$) but due to the contractions, $u$'s side in all minimum $(p,u)$-cuts in $G'$ contains $x$. (The existence of even a single minimum cut not containing $x$ implies that the latest cut also doesn't contain $x$ and the observation holds.)
Let us analyze how this could have happened, and reach a contradiction in all cases.
Recall that $G'$ is obtained by starting from $G_i$ and contracting all connected components $C_1,\ldots,C_k$ in $T_S \setminus V'$, and that each component $C_j$ is a minimum $(p,s_j)$-cut for some $s_j \in S$.

\begin{itemize}
\item If for all $j \in [k]$ the component $C_j$ is either contained or disjoint from $C^p_u$, then $C^p_u$ is not affected by the contractions and it is still a minimum $(p,u)$-cut in $G'$. Since $x \notin C^p_u$ we are done.
\item Otherwise, there are $t \geq 1$ components $C_1,\ldots,C_{j_t}$ corresponding to some nodes $s_1,\ldots,s_{t} \in S$, such that for all $j \in [t]$ the component $C_j$ is neither contained nor disjoint from $C^p_u$.
Note that $(C_j,V(G_i)\setminus C_j)$ is a minimum $(p,s_j)$-cut in $G_i$, but unlike $C^p_u$ it is not necessarily a latest cut, and therefore they may ``cross'' each other.
There are two subcases:
\begin{itemize}
\item If for some $j \in [t]$ $s_j \notin C^p_u$ then we get a contradiction to the fact that $C^p_u$ is latest, because $(S_u = C^p_u \setminus C_j, V(G_i) \setminus S_u)$ is a smaller latest cut and by Lemma~\ref{lem:latest_minus} it has the same value.
\item If for all $j\in [t]$ we have $s_j \in C^p_u$ then we do not get a contradiction to the fact that $C^p_u$ is latest, but we can still point at a minimum $(p,u)$-cut in $G'$ that does not contain $x$. By Lemma~\ref{lem:latest_union}, the cut $(S_u,V(G_i)\setminus S_u)$ where $S_u=C^p_u \cup \bigcup_{j=1}^t C_j$ is also a minimum $(p,u)$-cut in $G_i$.
Since $x \in V'$ it cannot be in any $C_j$ and therefore $x \notin S_u$. 
Finally, this cut is not affected by the contractions and therefore it is also a minimum cut in $G'$.
\end{itemize}
\end{itemize}

\end{proof}

\paragraph{Few problematic vertices with respect to $p$.}
The second important property that the pivot $p$ has with high probability, that follows from the shallowness of $\TG_p$, is that for all but a few \emph{problematic} nodes $v \in V'$ the latest minimum $(p,v)$-cut in $G'$ has $\leq n'/2$ nodes in $V'$. 
In particular, a node is in the set of problematic nodes $P = \{ u \in V' \mid |C_u^p \cap V'|>n'/2 \}$ if its side is larger than $n'/2$ in any minimum cut to $p$.
The issue with problematic vertices is that even when we find their latest minimum $(p,v)$-cut we cannot ``recurse on it'' if we want the recursion depth to only be logarithmic.

\begin{claim}
\label{cl:problematic}

The number of problematic vertices is bounded by $r/2$ with probability at least $1-1/N^{\gamma}$.

\end{claim}

\begin{proof}

All nodes in $P$ must be along a single path in $\TG_p$. There cannot be a divergence since the subtree of each of them has weight $>n'/2$ (and the total weight is $n'$).
By Claim~\ref{cl:depth}, all paths in $\TG_p$ have weight at most $r/2$ and therefore there can be $\leq r/2$ nodes in $P$.
\end{proof}
\paragraph{All are done after $O(\log{n})$ iterations}
We are now ready to start analyzing the calls to the \expanders procedure.
Recall that all $v \in V'$ have estimates $c'(v)$ for the minimum $(p,v)$-cut value, and actual cuts testifying for these values.
We say that a node is \emph{done} if $c'(v) = \lambda_{p,v}$, and is \emph{undone} if $c'(v)>\lambda_{p,v}$.
Since initially $c'(v)=\deg_G(v)$, only nodes that are leafs in $T'$ are done.
The (contracted) nodes in $V(G')\setminus V'$ are treated as if they are done.

The crux of the analysis of the \expanders method is the proof of Claim~\ref{cl:main_done}, stating that all undone nodes $v \in V'$ that have the following property at the beginning of an iteration become done after it.
The property is that there exists a minimum $(p,v)$-Cut $(C_v,V(G')\setminus C_v)$ in $G'$ such that the number of undone nodes $u \in C_v$ is at most $r$.
(Equivalently, a node satisfies the property if its latest minimum $(p,v)$-cut in $G'$ contains up to $r$ undone nodes.)
Moreover, by Claim~\ref{cl:main_done}, the testifying cuts for the done nodes are the latest minimum $(p,v)$-cut in $G'$.

The following claim builds on Claim~\ref{cl:main_done} to prove that after $O(\log{n})$ iterations, \emph{all} nodes in $V'$ become done.
To do this we utilize our structural understanding of the cut-membership tree $\TG_p$ that we formalized in Claim~\ref{cl:depth}, namely that the weighted depth is at most $r/2$.
In particular, we use the important corollary that all bags $B \in V(\TG_p)$ have weight $w(B) \leq r/2$ with high probability.

The intuition is as follows. Initially, all leaves are done. Then, after the first iteration, all nodes right above the leaves also become done, and moreover, any node whose subtree has $\leq r/2$ nodes must become done. Then, we can ``remove'' all done nodes from consideration and repeat the same reasoning.
How many times will it take until all are done? 
This is a simple question of analyzing a certain process on trees.
There are two extremes: either $\TG_p$ is path-like, but since the weighted-depth is $\leq r/2$ the whole path will be done in one iteration, or it is like a binary tree, but then we make a lot of progress in parallel and $O(\log n)$ iterations suffice to remove one level at a time.

\begin{claim}
\label{cl:repetitions}
After $\log n +2$ repetitions of the \expanders procedure, all nodes are done with probability at least $1-1/N^{\gamma-2}$. Moreover, the testifying cut for all $v \in V'$ is the latest minimum $(p,v)$-cut.
\end{claim}

\begin{proof}
An iteration of the \expanders procedure is considered \emph{successful} if all nodes $v$ with at most $r$ undone cut-members become done. 
By Claim~\ref{cl:main_done} and a union bound over all $n'\leq N$ nodes in $V'$, an iteration is successful with probability $\geq 1-1/N^{\gamma-1}$.
By another union bound, we get that the probability that all $\log n +2$ iterations are successful is at least $\geq 1- 2\log{n}/N^{\gamma-1}$.
The rest of this proof assumes that all iterations are successful, and moreover that the weighted-depth of $\TG_p$ is $\leq r/2$. 
By Claim~\ref{cl:depth}, the latter happens with probability $1-1/N^{\gamma}$; therefore, the conclusion of the proof holds with probability at least $1-1/N^{\gamma-2}$.

A bag $B\in V(\TG_{p})$ is said to be \emph{done} if and only if \emph{all nodes} in the bag and in its entire subtree in $\TG_p$ are done. 
Denote by $D^{(j)}$ the set of done bags in $V(\TG_{p})$ after the $j^{th}$ iteration of the \expanders procedure.
And let $T^{(j+1)} = V(\TG_{p}) \setminus D^{(j)}$ be the set of bags in $\TG_{p}$ that are not in $D^{(j)}$. 
Intuitively,  $T^{(j)}$ is the subtree of $\TG_{p}$ that we still have to handle (the set of undone bags) after iteration $j-1$, that is, when iteration $j$ begins.

Our goal is to show that for all $j\geq 1$, either $w(D^{(j)} \setminus D^{(j-1)}) \geq  w(T^{(j)})/2$, meaning that we make a lot of progress in each iteration, or all nodes are done in iteration $j+2$. It follows that $ \log{n}+2$ iterations suffice to make all nodes done, i.e. $D^{(\log{n}+2)}= V(\TG_{p})$. 

Let $j\geq 1$. To simplify notation, let $D=D^{(j)} \setminus D^{(j-1)}$ be the set of newly done bags in this iteration; our goal is to lower bound $w(D)$. Let $z := \frac{w(T^{(j)})}{r/2}$, and separate the analysis into three cases as follows:
\begin{itemize}
\item If $z \leq 2$ then all nodes are done after two additional iterations because in every iteration, either all bags become done, or the remaining weight is reduced by at least $w(D) \geq  r/2$. 
To see this, suppose that there is a bag that is not done after iteration $j$. Identify one such bag $B_v$ without other undone bags in its subtree (e.g. the lowest undone bag), and let $v$ be one of the undone nodes in this bag. By Claim~\ref{cl:main_done}, $v$ has at least $r$ undone cut-members at the beginning of iteration $j$. The nodes in $C_v$, the cut-members of $v$, can either be in $v$'s bag (that contains only $w(B_v)\leq r/2$ nodes), or in the rest of its subtree $C_v \setminus B_v$ (and are therefore done after iteration $j$). Therefore, there are at least $ r - r/2$ undone nodes in $C_v \setminus B_v$. Let $X$ be the set of bags containing these nodes. Clearly, $w(X) \geq r/2$. To conclude this case, note that $w(D) \geq w(X)$ because all bags in $X$ are undone before but done after iteration $j$ since they are in the subtree below $B_v$.
  
\item If $z > 2$ we have two more cases.
 Denote by $Q$ the set of bags in $T^{(j)}$ whose subtree in $T^{(j)} \cap \TG_{p}$ has total weight $> r$. 
 Note that all nodes not in bags in $Q$ have at most $r$ undone cut-members, since all nodes outside $T^{(j)}$ are done.
 Therefore, by Claim~\ref{cl:main_done}, all bags not in $Q$ are done after the $j^{th}$ iteration, and therefore $T^{(j)} \setminus Q \subseteq  D$.

\begin{itemize}

\item If $w(Q)< w(T^{(j)})/2$, then $w(D) \geq w(T^{(j)} \setminus Q ) = w(T^{(j)}) - w(Q) \geq w(T^{(j)})/2$ and we are done.

\item The slightly trickier case is when $w(Q) \geq w(T^{(j)})/2$.
Let $L_Q$ be the set of leaves in the subtree of $\TG_{p}$ induced by $Q$, i.e. the set of bags in $Q$ whose children in $\TG_{p}$ are not in $Q$.

First, we claim that $|L_Q|\geq z/2$, because:
(1) $w(Q) \leq |L_Q| \cdot r/2$ since the weight-depth of $\TG_{p}$ is at most $r/2$, 
(2) by the assumption of this case $ w(T^{(j)})/2\leq w(Q)$, and
(3) $ w(T^{(j)}) = z \cdot r/2 $ by definition of $z$.
Combining all three we get:
$$
|L_Q| \geq w(Q) / (r/2) \geq  w(T^{(j)})/r \geq z/2.
$$

Now, each bag $B$ in $L_Q \subseteq Q$ has a subtree $T_B \subseteq T^{(j)} \cap \TG_{p}$ of weight $w(T_B) > r$. 
Since $w(B) \leq r/2$ the weight below $B$ is at least $w(T_B \setminus \{B\}) > r-r/2=r/2$.
Let $X$ be the set of all bags in $T^{(j)} \cap \TG_{p}$ below the bags in $L_Q$, i.e. $X = \bigcup_{B \in L_Q}  T_B \setminus \{B\} $.
Since the subtrees of bags in $L_Q$ are disjoint, we get that $w(X) > |L_Q| \cdot r/2 \geq  zr/4 = w(T^{(j)})/2$.
Finally, observe that $w(D) \geq w(X)$ because all bags in $X$ are in $T^{(j)} \setminus Q \subseteq D$.
 
\end{itemize}

\end{itemize}

The fact that all testifying cuts are \emph{latest} follows because by Claim~\ref{cl:main_done}, whenever a node $v \in V'$ becomes done in the \expanders procedure we also have its latest minimum $(p,v)$-cut.

\end{proof}


\section{The Expander-Guided Querying Procedure}
\label{sec:expanders}

Suppose we have an unweighted graph $G$ on $N$ nodes that has already undergone preprocessing, and for all $k \leq w\leq n$ it has sparsifiers $G_w$ and expander decompositions $H^{w}_1,\ldots,H^w_h$ with some parameter $\phi_w$.
Moreover, suppose we have a \GHEPT $T$ for $G$ and let $V'$ be one of its super-nodes with its auxiliary graph $G'$ and a designated pivot node $p \in V'$.
Denote $|V'|=n'$, $|V(G')| = n$, and $|E(G')|=m$.
Also suppose that $T$ is a refinement of a $k$-Partial Tree of $G$ and therefore $\lambda_{u,v} > k$ for all $u,v \in V'$.
All nodes $v \in V'$ have estimates $c'(v)$ for the minimum $(p,v)$-cut value, and actual cuts testifying for these values.
We say that a node is \emph{done} if $c'(v) = \lambda_{p,v}$.
Moreover, suppose we have a \emph{room} parameter $0\leq r \leq n$  satisfying $r \geq 1/\phi_w$.

The goal of this section is to give an algorithm we call \emph{\expanders Procedure} such that any node $v \in V'$ will become done, if its latest minimum $(p,v)$-cut in $G'$ contains up to $r$ undone nodes. Moreover, the testifying cut we find for $v$ will be the latest minimum $(p,v)$-cut in $G'$.

\paragraph{The \expanders Procedure}

For each $w=2^j$ where $\lfloor \log{k} \rfloor \leq  j \leq \lceil \log{N} \rceil$ we do the following. These steps will target undone vertices $v$ that have $w \leq \lambda_{p,v} < 2w$. Note that any such node must also have $\deg_{G}(v)\geq c'(v) > \lambda_{p,v} \geq w$, and therefore only nodes with $\deg_{G} > w$ will be relevant. 
Let us define a node $u$ to be \emph{$w$-relevant} if $\deg_{G}(u)>w$.

\begin{enumerate}
\item First, we compute a Nagamochi-Ibaraki sparsifier $G'_w$ of our auxiliary graph with $O(n w)$ edges and such that all cuts of value $< 2w$ are preserved, and all larger cuts still have value $\geq 2w$. 
We perform our \MF queries on the sparsifier, making the running time of each query $MF(n,\min(nw,m))$ rather than $MF(n,m)$, and, since we anyway disregard cuts of value $\geq 2w$, it is equivalent to performing the queries on the original auxiliary graph.

\item As a preparatory step, consider the expanders $H^w_1,\ldots,H^w_h$ and for each expander $H$ compute: the set of $w$-relevant nodes $\hat{H}=\{ u \in V' \cap H \mid \deg_{G}(u)>w \}$ and their number $\alpha(H) = | \hat{H} | $.
We stress that in the description below, the size $|H|$ and number of outward edges $\delta(H)$ of an expander are defined with respect to $V(G)$ and $E(G_w)$ not with respect to $V'$ and $E(G')$.

\item\label{step:guess_xy} For each $x,x' \in \{ 2^j \}_{j=0}^{\lceil \log{N} \rceil}$ and  $y,y' \in \{0\} \cup \{ 2^j \}_{j=0}^{\lceil \log{N} \rceil}$, do the steps below. It is helpful to imagine that we are searching for an undone node $v$ that exists in some expander $H$ and to think of $x',y'$ as two guesses for the sizes of the two sides of $v$'s optimal cut when projected onto $H$.\footnote{More specifically, $x,x'$ will represent the side containing $v$ and $y,y'$ will represent the other side; this is why we let $y,y'$ be zero but $x,x'\geq 1$.}
The guesses $x,y$ are similar, but are only concerned with the number of $w$-relevant nodes on each side.

\begin{enumerate}
\item\label{step:skip} Skip any iteration where $x>x'$ or $y>y'$ or $\min\{x,y\} > 2/ \phi_w$. 

\item\label{step:large_x} If $x' \geq w/8$ or $y'\geq w/8$ then we are in a \emph{large H} case. Take each of the $O(N/w)$ expanders $H$ in the decomposition of $G_w$ with $|H|\geq x' + y'$ and do either of the following two operations for each:
\begin{enumerate}

\item \label{small_x} If $x \leq 2/ \phi_w$, do the sub-procedure:
\begin{enumerate}
\item \label{solve_x} Repeat the following $4\gamma x \ln{n}$ times: Construct a set $C \subseteq V'$ by adding each node $v\in \hat{H}$ to $C$ with probability $1/2x$, call the procedure in Theorem~\ref{thm:proc_main} on $G'_w,p,C$ to get a $(p,v)$-cut for all $v\in C$ and update the estimates $c'(v)$ and $S_v$ if the new cut has smaller value (and if their value is not $2w$). 
\end{enumerate}

\item\label{solve_y} Otherwise, it must be that $y \leq 2/ \phi_w$. Let $P$ be the set of $5 r$ nodes $v \in \hat{H}$ with highest $c'(v)$ breaking ties arbitrarily. Perform a \LMwC query to get the latest minimum $(p,v)$-cut in $G_w'$ for each $v \in P$, and (unless the returned cut has value $\geq 2w$) update the estimates $c'(v)$ and $S_v$.\footnote{Note that in this case, $v$ must become done unless $\lambda_{p,v}>2w$.}

 \end{enumerate}
  
\item\label{step:small_x} If $x'< w/8$ and $y' < w/8$ then we are in a \emph{small $H$} case. Consider all expanders $H$ (in the decomposition of $G_w$) with $\alpha(H) \leq 2(x+y)$ and $\delta(H) \geq (x+y) \cdot w/2$, i.e. expanders with relatively many edges leaving them.
Perform a \LMwC query to get the latest minimum $(p,v)$-cut in $G'_w$ for each $v \in \hat{H}$, and (unless the returned cut has value $\geq 2w$) update the estimates $c'(v)$ and $S_v$.

\end{enumerate}

\end{enumerate}

\paragraph{Running Time Analysis}
We will prove the following Lemma.
\begin{lemma}\label{Lemma:Bounds}
The time it takes to run the \expanders Procedure is:
\begin{enumerate}
\item\label{Bound:Existing:Algebraic:Densest:Lemma}
$\tO(\sqrt{N}\cdot \sqrt{M} \cdot n + 
N^{3/2}/\sqrt{M} \cdot m +
N^{3/2}/\sqrt{M} \cdot n^{3/2})$ 
 where $r=\sqrt{M/N}$ and $\phi_w = \sqrt{N/M}$,

\item\label{Bound:Existing:Combinatorial:Densest:Lemma}
$\tilde{O}\left(  N^{7/6}/\sqrt{M} \cdot m\sqrt{n} + 
N^{5/6}\cdot\sqrt{M}\cdot n  \right)$ using only combinatorial methods and where $r= \sqrt{M}/N^{1/6}$ and $\phi_w = \sqrt{M}/(N^{1/6}\cdot w)$.

\end{enumerate}
\end{lemma}

\begin{proof}
Let us start with analyzing Step~\ref{step:guess_xy} since it is the most interesting and expensive.
There are $O(\log^4{N})$ guesses for $x,y,x',y'$ and for each one we perform the following sub-cases.
Let $MF(n,m,f)$ be the upper bound for Max-Flow with unit capacity edges on graphs with $n$ nodes, $m$ edges, and flow size that is bounded by $f$. Furthermore, define $\beta:=\min\{m,nw\}$, and $\beta':=\min\{M,Nw\}$ (these are the numbers of edges in the graphs we are working with, due to the sparsification).

\begin{itemize}
\item In Case~\ref{solve_x} there are $\tilde{O}(1/ \phi_w)$ calls to the \pC procedure for each of the $O(N/w)$ large expanders, giving a total of $\tilde{O}(\frac{N}{w \phi_w})$ calls. Each call costs $\tilde{O}(MF(n,\beta,F_{a}))$ time for some $F_a\leq m$, giving a total time of 
$\tilde{O}\left(N/(w\phi_w)\cdot MF(n,\beta,F_{a})\right)$.

\item In Case~\ref{solve_y} there are $O(r)$ calls to
$MF(n,\beta,F_{b})$ for some $F_{b}\leq 2w$, for each of the
$O(N/w)$ large expanders, giving a total time of $\tilde{O}\left( Nr/w \cdot MF(n,\beta,F_{b})\right)$.

\item In Case~\ref{step:small_x}, we handle small $H$'s and we cannot upper bound their number based on how many nodes they contain. Instead, we get a bound based on the number of \emph{edges} leaving these expanders. Indeed, we consider only expanders $H$ with $\delta(H) \geq (x+y) \cdot w/2$ and perform only $O(x+y)$ calls to $MF(n,\beta,F_{c})$ for each one, where $F_{2w}\leq 2w$.
Since $\sum_{i=1}^h \delta(H^w_i)=\tilde{O}(\beta'\cdot \phi_w)$ for $\beta'=\min\{M,Nw\}$, we know that there are only $\tilde{O}(\frac{\beta'\cdot\phi_w}{(x+y)w})$ such $H$'s. 
Thus, the total number of queries is $\tilde{O}(\frac{\beta'\phi_w}{(x+y)w} \cdot (x+y)) = \tilde{O}(\beta' \phi_w/w)$.
The time for each query is $MF(n,\beta,F_{c})$ time for some $F_c\leq 2w$, giving a total time of 
$\tilde{O}\left( \beta'\phi_w/w\cdot MF(n,\beta,F_c)\right)$.

\end{itemize}

The first two steps in the procedure of computing the sparsifier $G_w'$ and precomputing the $w$-relevant subset of each expander only take $O(m)$ and $O(N)$ time, respectively.
Finally, there is an additional $O(\log{N})$ factor since everything is repeated for all the $w$'s. 

The total running time of the algorithm is a function of the \MF procedures that we utilize, and the parameters $r$, $\beta=\min\{m,nw\}$, $\beta'=\min\{M,Nw\}$, and $\phi_w$. 
The \MF algorithm that we use depends on the bound on the flow.
The running time of the \expanders Procedure is:

$$
\tilde{O}\left( N/w \left( 1/{\phi_{w}}\cdot MF(n,\beta,F_{a}) + r \cdot MF(n,\beta,F_{b}) \right) + 
\beta' \cdot \phi_w/w \cdot MF(n,\beta,F_{c})\right).
$$
\begin{enumerate}
\item\label{Bound:Existing:Algebraic:Densest}
To get the bound we set the following. 
\begin{enumerate}
\item
$MF(n,\beta,F_{a})=\tO(\beta+n_i^{3/2})$ (by the very recent algorithm~\cite{linearflow21}), which is at most $\tO(n_i w)$ since $w\geq k = \sqrt{N}\geq \sqrt{n_i}$,
\item
$MF(n,\beta,F_{b})=\tO(n_i w)$ (by the Karger-Levine algorithm~\cite{KL15}),
\item
$MF(n,\beta,F_{c})=\tO(\beta+n_i^{3/2})$.
\end{enumerate}
Furthermore, $\beta'\leq Nw$ and we use the given $r,\phi$. The bound is thus at most:
$$
\tO(N/w \cdot ( \sqrt{M/N} \cdot \tO(nw) + \sqrt{M/N} \cdot \tO(nw) + N\cdot \sqrt{N/M} \cdot (m+n^{3/2}))
$$
$$
=
\tO(\sqrt{N}\cdot \sqrt{M} \cdot n + 
N^{3/2}/\sqrt{M} \cdot m +
N^{3/2}/\sqrt{M} \cdot n^{3/2})
$$
concluding item~\ref{Bound:Existing:Algebraic:Densest:Lemma} in Lemma~\ref{Lemma:Bounds}.
\item\label{Bound:Existing:Combinatorial:Densest}
To get the bound for \emph{combinatorial} algorithms, we can set $MF(n,\beta,F_{a})=\tO(m\cdot n^{2/3})$
(by the Goldberg-Rao algorithm~\cite{GR98}), $MF(n,\beta,F_{b})=MF(n,\beta,F_{c})=\tO(n\cdot w)$ (by the Karger-Levine algorithm~\cite{KL15})
and the given $r,\phi$.
The bound simplifies to:
$$
\tilde{O}\left(N/w \cdot ( mn^{2/3}\cdot N^{1/6} \cdot w/\sqrt{M} + \sqrt{M}/N^{1/6}\cdot nw ) + 
N\sqrt{M}/(N^{1/6}\cdot w) \cdot nw \right) 
$$
$$
= \tilde{O}\left( N^{7/6}/\sqrt{M} \cdot mn^{2/3} + 
N^{5/6}\cdot\sqrt{M}\cdot n +
N^{5/6}\cdot\sqrt{M}\cdot n \right)  
$$
$$
\leq
\tilde{O}\left(  N^{7/6}/\sqrt{M} \cdot mn^{2/3} + 
N^{5/6}\cdot\sqrt{M}\cdot n  \right),
$$
concluding item~\ref{Bound:Existing:Combinatorial:Densest:Lemma} in Lemma~\ref{Lemma:Bounds}.
\end{enumerate}
\end{proof}

\paragraph{Correctness Analysis}
Recall that a node $v\in V'$ is \emph{done} if $c'(v)=\lambda_{p,v}$ and is \emph{undone} if $c'(v)>\lambda_{p,v}$.
The (contracted) nodes in $V(G')\setminus V'$ are treated as if they are done.

The following is the central claim, proving that the expander-based method works well for nodes that only have few undone cut-members. This will be sufficient for showing that all nodes become done after a few iterations in Claim~\ref{cl:repetitions}.

\begin{claim}
\label{cl:main_done}
If $v \in V'$ is an undone vertex such that there exists a minimum $(p,v)$-cut $(C_v,V'\setminus C_v), v\in C_v$ in $G'$ with at most $r$ undone cut-members $u \in C_v$ at the beginning of the \expanders procedure where $r \geq 1/\phi_w$, then $v$ is done at its end with probability at least $1-1/n^{\gamma}$.
Moreover, the testifying cut is the latest minimum $(p,v)$-cut in $G'$.
\end{claim}

\begin{proof}
Let $w$ be the power of $2$ such that $w\leq \lambda_{p,v} <2w$, and let $H^w_j$ be the expander that contains $v$ in the expander decomposition of $G_w$ with parameter $\phi_w$. For simplicity of notation, let us denote this expander by $H$.

Let $(C_v,V'\setminus C_v), v\in C_v$ be the minimum $(p,v)$-cut in $G'$ with at most $r$ undone cut members, as in the statement.
Let $\bar{C_v} \subseteq V$ be the set of nodes that are either in $C_v$ or in contracted nodes inside $C_v$. 
More formally, $\bar{C_v} = (C_v \cap V' ) \cup \{ v \in X \mid X \in C_v \cap (V(G') \setminus V')  \}$.
That is, $\bar{C_v}$ is the set that we get if we start from $C_v$ and ``uncontract'' all contracted nodes in $V(G')\setminus V'$.
Since $V'$ is a super-node in a \GHEPT and since $(C_v, V' \setminus C_v)$ is a minimum $(p,v)$-cut in its auxiliary graph, we know that $(\bar{C_v},V \setminus \bar{C_v})$ is also a minimum $(p,v)$-cut in $G$ of value $\lambda_{p,v}$. 

Now let us look at the projection of this cut onto $H$; define $L=\bar{C_v} \cap H$ and $R=(V \setminus \bar{C_v}) \cap H$ as the two sides of this cut in $H$, and note that $v \in L$ and that $(L,R)$ defines a cut in $H$ since $R=H \setminus L$.

Recall the definition of $\hat{H}$ as the set of $w$-relevant nodes in $H$ and similarly define:
$$
\hat{L} = \{ u\in L \mid \deg_{G}(u)>w \}, \hat{R} = \{ u\in R \mid \deg_{G}(u)>w \}.
$$ 
Observe that since $v$ is undone, $\deg_{G}(v)>\lambda_{p,v}$, and therefore $v \in \hat{L}$, i.e. $v$ is $w$-relevant.
Moreover, let $x,x',y,y'$ be the powers of $2$ that approximate the number of nodes and of $w$-relevant nodes in $L, R$: 
$$
x\leq |\hat{L}|<2x, y\leq |\hat{R}|<2y, x'\leq |L|<2x', y'\leq |R|<2y'. 
$$
In the corner cases where $R=\emptyset$ or $\hat{R}=\emptyset$ or both, we define $y:=0$ or $y':=0$ or both. 

To prove the claim, we will argue that the iteration of Step~\ref{step:guess_xy} that corresponds to $w,x,x',y,y'$ will successfully compute an optimal cut for $v$.
Notice that all \LMwC queries and calls to the \pC procedure inside this step will be performed on $G_w'$ rather than $V'$, but since $\lambda_{p,v}<2w$ this will not affect us: the latest minimum $(p,v)$-cut in $G_w'$ is the same as the latest minimum $(p,v)$-cut in $V'$.

First, we claim that the iteration will not be skipped in Step~\ref{step:skip} because $x\leq x'$, $y\leq y'$, and $\min\{x,y\} \leq 2/\phi_w$.
The first two follow by definition, because $\hat{L} \subseteq L$ and $\hat{R} \subseteq R$. The third inequality follows from the following reasoning. 
Since $H$ must have expansion $\Phi_{G_w\{ H \}} \geq \phi_w$, any of its cuts, including $(L,R)$, must have conductance:
$$
\Phi_{G_w\{ H \} }(L) = \frac{\delta(L,R)}{\min\{\vol_{G_w\{ H \}}(L), \vol_{G_w\{ H \}}(R)\}} \geq \phi_w.
$$
A key observation is that $\delta(L,R)<2w$ 
because any edge between $L$ and $R$ is also an edge between $C_v$ and $V\setminus C_v$, and the total weight of the latter is $<2w$.
Therefore, we get that 
$$
\min\{\vol_{G_w}(L), \vol_{G_w}(R)\} \leq \delta(L,R)/\phi_{w} < 2w / \phi_w.
$$
Note that we have changed the subscripts since $\vol_{G_w\{ H \}}(S)=\vol_{G_w}(S)$ for all $H,S$.
Next, observe that $\vol_{G_w}(L) > x \cdot w$ 
since $L$ contains at least $x$ nodes with $\deg_{G}> w$, and by the properties of our sparsifier for all nodes $\deg_{G_w}(u)>w$ if and only if $\deg_{G}(u)> w$.
For a similar reason $\vol_{G_w}(R) > y \cdot w$.
It follows that $\min\{x, y\} < \min\{\vol_{G_w}(L)/w, \vol_{G_w}(R)/w\} < 2/\phi_w$.

Once we have skipped Step~\ref{step:skip}, we will do one of the following steps, based on the values of $w,x,x',y$:
\begin{itemize}
\item If $x' \geq w/8$ or $y' \geq w/8$ then since $|H| \geq x'+y'$ by our definition of $x',y'$, $H$ will be considered in Step~\ref{step:large_x} where one of the Steps~\ref{solve_x} or~\ref{solve_y} will be performed on it: 
\begin{itemize}
\item If $x \leq 2 /\phi_w$, then Step~\ref{solve_x} is performed; let us analyze what happens.
There is a randomized process that gets iterated $4\gamma x \ln{n}$ times, that involves picking a subset $C \subseteq V'$ and calling the \pC procedure in Lemma~\ref{thm:proc_main}.
We say that subset $C$ is $v$-successful if $C \cap C_v =\{v \}$, i.e. $v$ is the only node who was chosen among its cut-members.
First, observe that if we call the \pC procedure on $G_w'$ with pivot $p$ and connected subset $C$, where $C$ is $v$-successful, then by Lemma~\ref{thm:proc_main} we will get the latest minimum $(p,v)$-cut in $G'_w$ (and therefore in $G'$) and we are done.
To conclude the proof of this case, it remains to argue that in at least one of the $4\gamma x \ln{n}$ iterations $C$ is $v$-successful with high probability. 
Consider one of the iterations. Since $v \in \hat{H}$ the probability that it gets chosen to $C$ is exactly $1/2x$. 
The only cut-members of $v$ that can get chosen are those in $\hat{L}$, because we only choose nodes in $\hat{H}=\hat{L} \cup \hat{R}$ and the nodes in $\hat{R}$ are not cut-members of $v$.
The expected number of cut-members in $C$ is therefore $|\hat{L}|/2x < 1$, and using Markov's Inequality, the probability that it is at least $2$ is $< 1/2$.
Therefore, the probability that $C$ is $v$-successful, i.e. both $v \in C$ and $|C|<2$, is at least $1/2x \cdot 1/2 = 1/4x$.
And the probability that at least of one the iterations produces a $v$-successful set $C$ is $\geq 1 - (1-1/4x)^{4\gamma x\ln{n}} \geq 1-1/n^{\gamma}$.
Note that for the analysis of this case, we did not use the fact that $v$ only has $r$ undone cut-members, this will be used in the next case.

\item Otherwise it must be that $y \leq 2 /\phi_w$, and Step~\ref{solve_y} is performed, where we simply compute a \LMwC query between $p$ and each node in a certain set $P$. We will show that $v$ must be in $P$.
Recall that $P$ contains the $5 r$ nodes in $\hat{H}$ with highest $c'$.
Let us upper bound the number of nodes $u \in \hat{H}$ that could have $c'(u)\geq c'(v)$. There are three kinds of nodes in $\hat{H}$:
(1) nodes in $\hat{R}$ and their number is $y_1 <2y \leq 4 / \phi_w \leq 4 r$ since we assume that $r \geq 1/\phi_w$ for all $w\leq N$, (2) undone cut-members of $v$ and their number is $y_2\leq r$ (by the assumption in the statement), and (3) done cut-members of $v$ whose number may be large but their $c'$ must be smaller than $v$'s.
To see this, let $u$ be a done cut-member of $v$, and observe that $c'(u)=\lambda_{p,u} \leq \lambda_{p,v}$ (since $v$'s cut contains $u$, it gives an upper bound for its connectivity to $p$) and therefore $c'(v)>\lambda_{p,v} \geq c'(u)$.
Thus, there can only be $y_1+y_2 < 5 r$ nodes in $\hat{H}$ with $c'$ that is higher or equal to that of $v$, meaning that $v$ must be in $P$ and it will become done.

\end{itemize}
 \item Otherwise, if $x'< w/8$ and $y' < w/8$, we argue that $H$ will be among the expanders that we consider, i.e. it satisfies both $\alpha(H) \leq 2(x+y)$ and $\delta(H) \geq (x+y) \cdot w/2$.
Once we have that, we are done because Step~\ref{small_x} computes a \LMwC query between $p$ and each nodes in $\hat{H}$.

The first condition holds by definition of $x,y$. Verifying that $\delta(H) \geq (x+y) \cdot w/2$ is a bit more tricky\footnote{This is the central ingredient that fails when the graphs are \emph{weighted}.}: intuitively, we prove that small expanders must have a relatively large number of edges going out.
Consider the $\geq (x+y)$ nodes in $\hat{H}$.
Each of these nodes has $>w$ adjacent edges in $G_w$ but only $|H|-1 < 2 (x'+y') < w/2$ of them could stay inside $H$.
Therefore, the other endpoint of at least $w/2$ of these edges is outside $H$.
In total, there are $>(x+y)\cdot w/2$ edges contributing to $\delta(H)$ and we are done.
Note that here we crucially rely on the fact that $G$ (and therefore $G_w$) are simple graphs\footnote{If there are weights or parallel edges in $G$, all of the $>w$ edges of a $w$-relevant node $v \in \hat{H}$ could go to the same neighbor $u \in H$ and $\delta(H)$ may be small. Consequently, we cannot get a good upper bound in the small $H$ case.}, and it is important that $(x'+y')$ approximates $|H|$ the size of $H$ in $V(G)$, not in $V(G')$, since $G'$ may contain parallel edges (due to the contractions).

\end{itemize}

\end{proof}

\fi 

\section{Conclusion}

This paper presents the first algorithm with subcubic in $n$ running time for constructing a \GHT of a simple graph and, consequently, for solving the \APMF problem.
It is achieved by a combination of several tools from the literature on this problem, as well as two new ingredients: the \expanders and \pC procedures.
The new ideas are reminiscent of recent algorithms \cite{KThorup19,saranurak2020simple,LP20} for the easier problem of \GMC.
We conclude with some remarks and open questions.

\begin{itemize}

\item The assumption that the graph is unweighted is only used in one specific case of the analysis for observing that: if a high degree node is in a small component then most of its edges must leave the component.
A similar observation is at the heart of the breakthrough deterministic \GMC algorithm of Kawarabayashi and Thorup \cite{KThorup19} but can now be avoided with a moderate loss in efficiency \cite{LP20}.
Thus there is room for optimism that $n^{1-\eps}\cdot T_{\MF}(n,m)$ time for \emph{weighted graphs} is possible with the available tools.

\item The new subcubic algorithm uses randomness in multiple places and succeeds with high probability. 
All of the ingredients can already be derandomized (with some loss) using existing methods, except for one: the randomized pivot selection (see Complication 1 in Section~\ref{overview}).
It is likely that a fully deterministic algorithm making $n^{1-\eps}$ queries is attainable but there is an inherent challenge that has also prevented the $\tilde{O}(n^3)$ time algorithm \cite{BHKP07} from being derandomized yet.
However, matching the new $\tilde{O}(n^{2.5})$ bound deterministically seems to require more new ideas, including a deterministic $\tilde{O}(n^2)$ algorithm for \MF in simple graphs that can be used instead of Karger-Levine \cite{KL15}.

\item Perhaps the most interesting open question is whether $\tilde{O}(m)$ time can be achieved, even in simple graphs and even assuming a linear-time \MF algorithm. 
The simplest case where breaking the $n^{2.5}$ bound is still challenging has been isolated in Section~\ref{sec:clcr}; perhaps it will lead to the first conditional lower bound for computing a \GHT?

\end{itemize}

\ifprocs
\begin{acks}
\else
\paragraph{Acknowledgements}
\fi 
We thank the anonymous reviewers for many helpful comments.
\ifprocs
\end{acks}
\fi

\ifprocs
\bibliographystyle{alphaurl}
\bibliography{robi}
\else
{\small
\bibliographystyle{alphaurlinit}
\bibliography{robi}
}
\fi

\appendix

\ifprocs\else 

\section{An Alternative Proof of the \pC Procedure}
\label{app:proc}

This section gives a complete proof of Lemma~\ref{thm:proc_main} by generalizing the arguments in Section~\ref{sec3:motivation}. 
The Lemma essentially follows from the very recent work of Li and Panigrahy \cite[Theorem II.2]{LP20} for \GMC. 
We provide another proof both for completeness and because it could be of interest as it exploits the structure of the \GHT instead of using the submodularity of cuts directly.
Let $MF(N,M,F)$ be an upper bound on \MF in graphs with $N$ nodes, $M$ edges, and where the flow size is bounded by $F$.

\begin{replemma}{thm:proc_main}[\textit{restated}]
Given an undirected graph $G=(V,E,c)$ on $n$ nodes and $m$ total edges, a \emph{pivot} node $p\in V$, and a set of \emph{connected} vertices $C \subseteq V$, let $(C_v,V\setminus C_v)$ where $v \in C_v, p \in V\setminus C_v$ be the latest minimum $(p,v)$-cut for each $v \in C\setminus\{p\}$.
There is a deterministic $O(MF(n,m,c(E)) \cdot \log{n})$-time algorithm that returns $|C|$ disjoint sets $\{ C'_v \}_{v \in C}$ such that
for all $v\in C$: if $C_{v} \cap C = \{v\}$ then $C'_v = C_v$.
\end{replemma}

We call a node $v\in C$ \lucky if it satisfies this condition.
We begin by describing the algorithm, then continue to its analysis.

\paragraph*{Algorithm.}
First, we prove this theorem with additional edges of very large 
capacities $U$ and $U^2$ for $U=c(E)^4$, and then we show that we actually do not need these edges.
%
\begin{enumerate}
\item
We define a recursive procedure $R(G,C,p)$ which operates as follows.
\begin{enumerate}
\item
If $|C|=1$, denote the node in $C$ by $u$, find $(S_{pu},S_{up}):=\LMC_{G}(p,u)$, and return $S_{up}$.
\item
Otherwise, if $|C|\geq 2$, do the following.
\begin{enumerate}
\item
Similar to Section~\ref{sec3:motivation}, connect $p$ with edges of capacity $U$ to every node in $C$, and denote by $G_m$ the resulting graph.
\item
Denote by $C_1$ an arbitrary set of $\lceil |C|/2 \rceil$ nodes from $C$, and by $C_2$ the remaining $\lfloor |C|/2 \rfloor$ nodes of $C$.
\item
Split $p$ into $p_1$ and $p_2$, connecting the newly added edges between nodes in $C_1$ and $p$ in $G_m$ to $p_1$, and the newly added edges between nodes in $C_2$ and $p$ in $G_m$ to $p_2$. Namely, for each edge $\{p,u\} \in E(G_m)$ where $u \in C_1$ we add the edge $\{p_1,u\}$ to $E(G_h)$ and for each edge $\{p,u\} \in E(G_m)$ where $u \in C_2$ we add the edge $\{p_2,u\}$ to $E(G_h)$.
\item
Add an edge of capacity $U^2$ between between $p_1$ and $p_2$, and denote the new graph $G_h$.
\item
Next, find $(S_{p_1p_2},S_{p_2p_1}=V\setminus S_{p_1p_2}):=\MC_{G_h}(p_1,p_2)$, where $p_1\in S_{p_1p_2}$ and $p_2\in S_{p_2p_1}$ (and thus also $C_1\subset S_{p_1p_2}$ and $C_2\subset S_{p_2p_1}$, since otherwise the weight of the cut would be at least $U^2+U$, while any cut with $C_1\subset S_{p_1p_2}$ and $C_2\subset S_{p_2p_1}$ has weight at most $U^2+c(E)<U^2+U$). 
\item
Continue recursively by calling $R(G_h(S_{p_1p_2}),C_2,p_2)$ and $R(G_h(S_{p_2p_1}),C_1,p_1)$ where $G_h(S_{p_1p_2})$ is $G_h$ after removing all edges of high capacities $U,U^2$, and contracting $S_{p_1p_2}$, and $G_h(S_{p_2p_1})$ is similar but after contracting $S_{p_2p_1}$.
\end{enumerate}
\end{enumerate}
\end{enumerate}

\paragraph*{Correctness.}
We begin by a description of the structure of every cut-equivalent tree $T_m$ of $G_m$, and then of every cut-equivalent $T_h$ of $G_h$ (see Figure~\ref{Figs:Procedure}).
\begin{claim}\label{Claim:T_m}
The following structural claims hold.
For every cut-equivalent tree $T_m$ of $G_m$ and every \lucky node $v\in C$, $T_m$ has an edge $\{p,v\}$ of capacity $U+\MF_{G}(p,v)$ between $p$ and $v$, such that removing $\{p,v\}$ from $T_m$ results in a partition of the nodes into the two subtrees $T_v$ and $T_p$ such that $(S_{pv}=V(T_p),S_{vp}=V\setminus S_{pv}=V(T_v))=\MC(p,v)$, where $v\in S_{pv}$ and $p\in S_{vp}$.
\end{claim}
\begin{proof}
This is simply true because the added edge of capacity $U$ must participate in every minimum $(p,v)$-cut in $G_m$, and that removing the edges of capacity $U$ from $G_m$ results in $G$.
\end{proof}
\begin{claim}\label{Claim:T_h}
For every cut-equivalent tree $T_h$ of $G_h$ and every node $v_1\in C_1$, $T_h$ has an edge of capacity $U+\MF_{G}(p,v_1)$ between $p_1$ and $v_1$, with a minimum $(p,v_1)$-cut attached to $v_1$, and symmetrically for every node $v_2\in C_2$, $T_h$ has an edge of capacity $U+\MF_{G}(p,v_2)$ between $p_2$ and $v_2$, with a minimum $(p,v_2)$-cut attached to $v_2$. 
In addition, $T_h$ must have $p_1$ connected to $p_2$ with an edge of capacity at least $U^2$. 
\end{claim}
\begin{proof}
This is true similarly to Claim~\ref{Claim:T_m}, and as $p_1$ is on the same side as $p_2$ in every minimum cut in $G_h$ that is not a minimum $(p_1,p_2)$-cut.
\end{proof}

\begin{figure*}[ht]
  \begin{center}
    \includegraphics[width=6.5in]{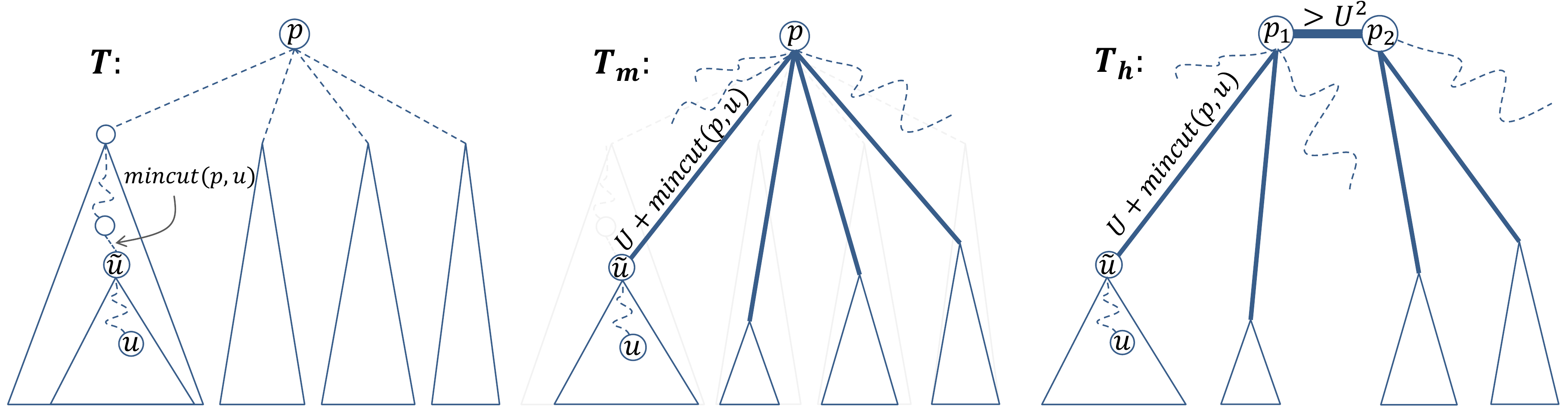}
    \end{center}
    \caption{An illustration of the trees $T,T_m, T_h$ in the procedure, where bold edges represent edges with added capacity $U$ and $U^2$. In this illustration, $\tilde{u}$ is a connected node.}
    \label{Figs:Procedure}
\end{figure*}

We will show the following claim for graphs $G_m$ along the execution.
\begin{claim}\label{Claim:technical_proc}
For every \lucky node $v$, if $v\in V(G_m)$ then $C_v\subset V(G_m)$ (i.e., the nodes of $C_v$ have not been contracted).
\end{claim}
Given claim~\ref{Claim:technical_proc}, it is easy to see that eventually $C_v$ will be contained in a recursive call that has exactly one connected node $v$, and there $C_v$ is a valid option for a $(p,v)$-cut. As a result, this cut will be returned, as required.

\begin{proof}[Proof of Claim~\ref{Claim:technical_proc}]
This is proved by induction on the recursion depth. The base case is depth $0$, namely the input graph right after the preprocessing step, that is $G_m$, which contains $C_v$ for every connected node $v$.
Let $G'_m$ be a graph along the execution.
Assuming that for a \lucky node $v$, it belongs to $V(G'_m)$, and we will show that $C_v$ is contained in $V(G'_m)$.
Consider the graph $G''_m$ in the previous depth upon which its $\MC(p_1,p_2)$, $G'_m$ was constructed.
Since $v\in V(G''_m)$, it holds by the induction hypothesis that $C_v\subset V(G''_m)$.
By the Claim~\ref{Claim:T_h}, $C_v$ must stick together in either side of $\LMC(p_1,p_2)$, and thus $C_v\subset V(G'_m)$ (this is true since $C_{v}\subset S_{vp}$ in every minimum $(p,v)$-cut $(S_{pv},S_{vp})$ in $G$), as required.
\end{proof}

\paragraph*{Running time.}
Note that after $O(\log |C|)$ iterations, we returned a cut for every node in $C$. Furthermore, for every depth in the recursion, the total number of  nodes in all graphs is bounded by $\tO(n)$, and of the edges by $\tO(m)$, and the total size of all flows queried in this depth is bounded by $\tO(c(E))=\tO(\sum_{e\in E}c(e)+|C|U+U^2)$, the latter is correct since the super-nodes in a single depth form a partition tree, and so by Claim $3.9$ in~\cite{AKT20} the total sum of all edge weights in these auxiliary graphs is bounded by $O(c(E))$. 

Finally, as we will show shortly, we will actually not use the edges with capacity $U,U^2$, and so the total flow queried simply becomes $\tO(c(E))=\tO(\sum_{e\in E}c(e))$.

\paragraph*{Lifting the assumption on heavy edges.}
Finally, we note that the added weights are not really necessary. The algorithm would be the same, except that we do not add edges of high capacity.
Instead, just before computing $\MC_{G_h}(p_1,p_2)$ we define a graph $\tilde{G}_h$ where we contract $p_1$ and $C_1$ into $p'_1$, and also $p_2$ and $C_2$ into $p'_2$, and then find $\MC_{\tilde{G}_h}(p'_1,p'_2)$. 
Recall that $(S_{p_1p_2},S_{p_2p_1})$ is $\MC_{G_h}(p_1,p_2)$.
Since in every $\MC_{G_h}(p_1,p_2)$ it holds that $C_1\subset S_{p_1p_2}$ and $C_2\subset S_{p_2p_1}$, $\MC_{\tilde{G}_h}(p_1,p_2)$ partitions $V$ in the same way as some $\MC_{G_h}(p_1,p_2)$ (considering super-nodes as the nodes they contain). This is because if we contract nodes on each side of a min cut but not across, then the same cut remains.

\fi 

\end{document}